\documentclass[10pt]{article}
\usepackage[subpreambles=true, sort=true]{standalone}
\usepackage[a4paper, lmargin=20mm, rmargin=20mm, tmargin=20mm, bmargin=21mm]{geometry}
\usepackage[utf8]{inputenc}
\usepackage[english]{babel}
\usepackage{mlmodern}
\usepackage[font=small]{caption}
\usepackage{bbm}
\usepackage{multicol}
\usepackage{titling}
\usepackage{titlesec}
\titleformat{\section}
  {\centering\Large}{\thesection}{1.1em}{}
\titleformat{\subsection}
  {\large}{\thesubsection}{.5em}{}
\titlespacing*{\section}{0pt}{4ex}{3ex}
\titlespacing*{\subsection}{0pt}{3ex}{2ex}
\usepackage{csquotes}
\usepackage{authblk}
\usepackage{amsmath, amsthm, amssymb, amsfonts, mathtools}
\usepackage[backend=biber, sorting=none, maxcitenames=10, url=false, isbn=false]{biblatex}
\AtEveryBibitem{
    \clearfield{urlyear}
    \clearfield{url}
}
\usepackage{hyperref}
\hypersetup{
    hidelinks,
    colorlinks=false,
    breaklinks=true,
    pdftitle={Gauging Q-Phases: MPS Approach}
}
\ifpdf
\else
  \usepackage[hyphenbreaks]{breakurl} 
\fi
\usepackage{eqparbox}
\usepackage{braket}
\providecommand{\dyad}[2]{{\ket{#1}\mkern-6mu\bra{#2}}}
\usepackage{graphicx}
\usepackage[shortlabels]{enumitem}
\setlist[enumerate,1]{label=(\roman*)}
\usepackage{float}
\usepackage[page]{appendix}

\usepackage{stackengine}
\usepackage{listofitems}
\newcommand\Stk[2]{\genfrac{}{}{0pt}{}{#1}{#2}}
\newcommand\mStk[2][\normalbaselineskip]{%
  \setstackgap{L}{#1}%
  \setsepchar{ }%
  \readlist\inputlist{#2}%
  \raisebox{#1*(\inputlistlen-1)/2}{\lstackMath\Longunderstack[c]{#2}}%
  \setstackgap{L}{\normalbaselineskip}%
}
\newcommand{\sunderbrace}[2]{\,\underbrace{\! #1 \!}_{#2}\,}
\newcommand{\Rep}{\operatorname{Rep}}
\newcommand{\Res}{\operatorname{Res}}
\newcommand{\Ind}{\operatorname{Ind}}
\newcommand{\cat}[1]{\mathrm{#1}}
\newcommand{\tw}[2]{\prescript{#1\mkern-4mu\relax}{}{#2}}
\newcommand{\twb}[2]{\prescript{#1\relax}{}{#2}}

\theoremstyle{definition}
\newtheorem{theorem}{Theorem}[section]
\newtheorem{theorem-no-proof}[theorem]{Theorem}
\newtheorem{definition}[theorem]{Definition}
\newtheorem{deflemma}[theorem]{Lemma and Definition}
\newtheorem{example}[theorem]{Example}
\newtheorem{corollary}[theorem]{Corollary}
\newtheorem{lemma}[theorem]{Lemma}
\AtEndEnvironment{example}{\null\hfill\ensuremath{\diamond}}
\AtEndEnvironment{theorem-no-proof}{\null\hfill\ensuremath{\square}}
\newtheorem*{remark}{Remark}
\newtheorem*{conjecture}{Conjecture}
\title{Gauging Quantum Phases: A Matrix Product State Approach}
\author[1]{David Blanik}
\author[2,3]{Jos\'e Garre-Rubio}
\author[1,2]{Norbert Schuch}
\affil[1]{University of Vienna, Faculty of Physics, Boltzmanngasse 5, 1090 Vienna, Austria}
\affil[2]{University of Vienna, Faculty of Mathematics, Oskar-Morgenstern-Platz 1, 1090 Vienna, Austria}
\affil[3]{Instituto de F\'isica Te\'orica, UAM-CSIC, C. Nicol\'as Cabrera 13-15, Cantoblanco, 28049 Madrid, Spain}
\date{}
\begin{document}
\renewcommand\Affilfont{\itshape\small}
\renewcommand{\abstractname}{\vspace{-4\baselineskip}}

\maketitle

\begin{abstract}
  \noindent Utilizing the framework of matrix product states,
  we investigate gauging as a method for exploring quantum phases of matter.
  Specifically, we describe how symmetry-protected topological (SPT) phases
  and spontaneous symmetry breaking (SSB) phases in one-dimensional spin systems
  behave under twisted gauging,
  a generalization of the well-known gauging procedure for globally symmetric states.
  Compared to previous order parameter--based approaches, our analysis is
  not limited to the case of maximally non-commutative (MNC) phases
  and we use our findings to propose a generalization of the Kennedy-Tasaki transformation
  to the non-MNC setting.
  A key result of our work is that gauging produces configurations
  characterized by a combination of MNC order and symmetry breaking,
  when applied to non-MNC SPT phases.
  More generally, we conjecture a precise correspondence between SSB
  and non-MNC SPT phases,
  possibly enabling the detection of such phases using local and string order parameters.
\end{abstract}

\smallskip

\begin{multicols}{2}
\section{Introduction}
\label{sec:introduction}
Dualities and gauging maps are pivotal to our understanding of quantum phases of matter,
particularly in systems exhibiting symmetries, where they establish connections between seemingly unrelated models, providing insight into underlying correspondences. For example, the Kramers-Wannier duality~\cite{kramersStatisticsTwoDimensionalFerromagnet1941} connects the disordered and ordered phases of the quantum Ising model, pinpointing its critical point.

In one-dimensional spin chains, certain symmetric phases, known as symmetry-protected topological (SPT) phases, have garnered increased attention, in part because they can be used as a resource in measurement-based quantum computation~\cite{miyakeQuantumComputationEdge2010, elseSymmetryProtectedPhasesMeasurementBased2012, millerResourceQualitySymmetryProtected2015, stephenComputationalPowerSymmetryProtected2017}.
These phases are protected by a global group symmetry and feature unique ground states, non-trivial edge modes and degeneracies in the entanglement spectrum. The most prominent duality in this context is the Kennedy-Tasaki (KT) transformation~\cite{kennedyHiddenSymmetryBreaking1992},
mapping the \(\mathbb{Z}_2\times \mathbb{Z}_2\)--symmetric
spin-1 Affleck-Kennedy-Lieb-Tasaki (AKLT) model~\cite{affleckRigorousResultsValencebond1987},
which belongs to the Haldane phase, to the fully symmetry broken phase,
revealing its ground state and edge mode structure.

Previous research has extended KT transformations to broader contexts, including \(\mathbb{Z}_N\times \mathbb{Z}_N\)-protected phases~\cite{duivenvoordenSymmetryprotectedTopologicalOrder2013} and maximally non-commutative (MNC) SPT phases for arbitrary abelian groups~\cite{elseHiddenSymmetrybreakingPicture2013}.
Both approaches rely on the fact that MNC phases are characterized by string order parameters \cite{pollmannDetectionSymmetryprotectedTopological2012}.
For the groups \(\mathbb{Z}_N\times \mathbb{Z}_N\), a non-unitary KT transformation generalized to closed spin chains has also been proposed~\cite{liNoninvertibleDualityTransformation2023}
and twisted gauging maps have been constructed~\cite{luRealizingTrialityPality2024} to study generalized \(p\)-alities and their permutation action on quantum phases.
For a generalization of the KT transformation to higher-integer-spin see~\cite{oshikawaHiddenZ2Z21992}.
Gauging as a source of dualities in more general settings has recently been investigated in \cite{lootensEntanglementDensityMatrix2025, cuiperTwistedGaugingTopological2025}.

It has long been realized that matrix product states (MPSs),
a family of one-dimensional tensor network states,
particularly prominent in the study of strongly correlated systems
because of their ability to efficiently approximate ground states of local, gapped Hamiltonians~\cite{hastingsAreaLawOnedimensional2007,wolfAreaLawsQuantum2008},
provide a convenient setting to classify 1D SPT phases \cite{pollmannDetectionSymmetryprotectedTopological2012,schuchClassifyingQuantumPhases2011}.
By leveraging the MPS framework,
where quantum phase information is encoded as local properties of the MPS tensors,
we are able to construct dualities in settings where existing approaches face challenges,
such as non-MNC SPT phases or systems deviating from renormalization group (RG) fixed points.

The results of the paper are as follows:
Based on \cite{haegemanGaugingQuantumStates2015}, we define a gauging procedure of states and a generalization thereof, which we call twisted gauging, and describe explicitly the emergent dual \(\Rep G\) symmetry of the gauged state. For abelian groups, we use this description to study the permutation of quantum phases under the application of the (twisted) gauging procedure, allowing us to view gauging as a mapping between phases.
In particular, we show that non-MNC SPT phases get mapped to phases featuring MNC SPT order in addition to symmetry breaking, a connection which could allow the detection of any SPT phase in terms of local and string order parameters.

In detail, we show that under non-twisted gauging the quantum phase
\((G_{\mathrm{u}}\leq G_{\mathrm{u}}\times G_{\mathrm{b}}, [\alpha]\in \operatorname{H}^2(G_{\mathrm{u}}, \operatorname{U}(1))\)
gets mapped to the phase
\begin{equation*}
  \left(
    \widehat{G_{\mathrm{b}}}\times \operatorname{ker}\Res ^{ G_{\mathrm{u}}}_{K_{\alpha}}\leq \widehat{G}, \left[1\cdot\varphi^{\star}\overline{\alpha}\right] \in \operatorname{H}^2(\widehat{G}_{\mathrm{u}}, \operatorname{U}(1))
  \right),
\end{equation*}
where the map \(\varphi\) depends on the specific choice of dual symmetry
and \(K_{\alpha}\) denotes the degeneracy (cf.\ Definition~\ref{def:degeneracy1}) of \([\alpha]\).
The twisted case follows an analogous pattern and is covered in the Main Theorem~\ref{thm:main}.
Our theorem is only applicable to the case where the unbroken symmetry group allows for a complement in~\(G\),
however, we provide examples suggesting that it continues to hold when no complement exists.

Furthermore, we are able to combine untwisted and twisted gauging in such a way that the mapping of phases mimics the expected behavior of a generalized KT transformation, also for non-MNC phases.
Finally, extrapolating from the Main Theorem, we propose a conjecture that links the degeneracy of the SPT phase with the broken subgroup of the symmetry-broken phase to which it is mapped under gauging, e.g.,
the trivial SPT phase is dual to the fully symmetry broken phase,
offering a deeper understanding of the underlying structure of such phases.

While our results are based on the observed mapping of states under the introduced gauging procedure, we emphasize that any MPS (under certain conditions explained below) is the ground state of a corresponding (parent) Hamiltonian, establishing our results also at the level of Hamiltonian models.

The behavior of quantum phases under twisted gauging
was recently also examined in \cite{luRealizingTrialityPality2024},
where conformal field theory methods were used to describe the mapping of
\(\mathbb{Z}_n\times \mathbb{Z}_n\)-symmetric phases under gauging.
Concerning on-site symmetries,
the results presented here are more general, since we cover all abelian symmetry groups,
only requiring that the unbroken subgroup admits a complement.
On the other hand, we do not consider the interplay between non-local symmetries,
e.g., translation invariance, and gauging of on-site symmetries.
Their definition of twisted gauging coincides with ours and in the cases where our results overlap they agree.

The present paper is organized as follows:
In Section~\ref{sec:tensor-networks}, we introduce matrix product states and matrix product operators (MPOs) as tools for describing one-dimensional quantum systems. We cover graphical notation, the classification of symmetry-protected topological and spontaneous symmetry breaking phases, and their representation in the MPS framework.

In Section~\ref{sec:gauging-procedure}, we introduce gauging, a procedure that transforms globally symmetric states into locally symmetric ones. We define untwisted and twisted gauging, formulate gauging as the application of an MPO, and analyze the emergent dual symmetry and entanglement structure of the resulting gauged MPS.

In Section~\ref{sec:mapping-phases}, we explore the effect of gauging on quantum phases in the case of finite abelian symmetry groups. We demonstrate that gauging permutes between SPT and spontaneous symmetry breaking (SSB) phases and establish a rigorous mapping of phases under gauging, including our Main Theorem that gives the precise transformation rules under twisted gauging.

\section{Basics of Tensor Networks and Graphical Notation in 1D}
\label{sec:tensor-networks}

\subsection{Matrix Product States and Operators}
\label{sec:mps-mpo}
A matrix product state (MPS)~\cite{perez-garciaMatrixProductState2007} is a type of tensor network state
used to efficiently describe quantum many-body systems,
especially ground states of local, gapped Hamiltonians in one-dimensional systems.
They are particularly useful in studying strongly correlated systems, e.g., quantum phase transitions.
MPSs can be visualized as a chain of nodes (tensors) connected by virtual bonds,
the contraction of which reconstructs the full quantum state.
\begin{definition}[MPS tensor]\label{def:mps}
  An \emph{MPS tensor} is an order-3 tensor 
  \(A \in \mathcal{V}_{\mathrm{o}}\otimes \mathcal{H} \otimes \mathcal{V}_{\mathrm{t}}^{*}\),
  where Hilbert spaces denoted as \(\mathcal{H}\) we shall refer to as \emph{physical},
  while Hilbert spaces denoted as \(\mathcal{V}\) we shall refer to as \emph{virtual}, in this context.
  Graphically, we will depict the MPS tensor \(A\) as
  \begin{equation}
    \label{eq:93}
    A \equiv \raisebox{-12.88516pt}{\includegraphics{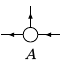}} \equiv \raisebox{-12.88516pt}{\includegraphics{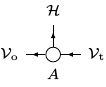}}
  \end{equation}
  and choosing bases for \(\mathcal{H}, \mathcal{V}_{\mathrm{o}}\) and \(\mathcal{V}_{\mathrm{t}}\),
  the tensor \(A\) can be viewed as a set of matrices \(\left\{ A^i \mid i = 1, \ldots, d_{\mathcal{H}} \right\}\) of size \(d_{\mathcal{V}_{\mathrm{o}}}\times d_{\mathcal{V}_{\mathrm{t}}}\).
\end{definition}
\begin{remark}[Graphical notation]
  Formally, in TN diagrams a Hilbert space is assigned to each edge (or leg)
  and tensors, depicted as nodes where multiple edges meet, are simply elements in the tensor product
  of the Hilbert spaces associated with its incident legs,
  where the dual of a Hilbert space appears in the tensor product,
  if the edge it is associated with terminates at the tensor.
  In graphical notation the tensor product of multiple tensors is depicted
  as drawing the factors next to each other.
  Composition, or more generally tensor contraction, is depicted as connecting the corresponding legs.
\end{remark}
\begin{definition}[Injective MPS tensor]\label{def:lri}
  An MPS tensor \(A\) is called
  \emph{injective}, if it is injective as a map 
  \begin{equation}
    \raisebox{-12.88516pt}{\includegraphics{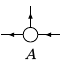}}\colon
    \mathcal{V}_{\mathrm{o}}^{*} \otimes \mathcal{V}_{\mathrm{t}} \longrightarrow \mathcal{H},
  \end{equation}
  going from virtual to physical Hilbert spaces.
  When interpreting TN drawings as linear maps, we shall usually view them as going from bottom to top
  and the orientation of each leg determines whether the associated Hilbert space or its dual appears
  in the (co)domain.
  Injectivity of the tensor \(A\) is equivalent to the existence of a tensor \(A^{-1}\),
  often called the \emph{(left) inverse} of \(A\), satisfying
  \begin{equation}
    \label{eq:95}
    \raisebox{-24.2663pt}{\includegraphics{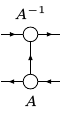}}=
    \raisebox{-9.08113pt}{\includegraphics{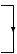}}\; \raisebox{-9.08113pt}{\includegraphics{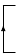}}
    \equiv \operatorname{id}_{\mathcal{V}_{\mathrm{o}}^{*}}
    \otimes \operatorname{id}_{\mathcal{V}_{\mathrm{t}}}.
  \end{equation}
  Another equivalent definition of injectivity is that the
  set of matrices \(\left\{ A^i \mid i = 1, \ldots, d_{\mathcal{H}} \right\}\)
  spans the entire space of 
  \(d_{\mathcal{V}_{\mathrm{o}}}\times d_{\mathcal{V}_{\mathrm{t}}}\) matrices.
\end{definition}
\begin{remark}
  For the rest of this section, if we do not specify otherwise, we will only consider MPS tensors with
  \(\mathcal{V}_{\mathrm{o}} = \mathcal{V}_{\mathrm{t}} \eqcolon \mathcal{V}\).
  In that case one can form a string of these tensors
  and contract the incoming leg of one with the outgoing leg of the next,
  thus producing a new tensor,
  a construction used in the following definition.
\end{remark}
\begin{definition}[Normal MPS tensor]
  An MPS tensor \(A\) is called \emph{normal} if there exists \(k\in\mathbb{N}\),
  such that the \emph{blocked} MPS tensor
  \begin{equation}
    \label{eq:85}
    \raisebox{-23.63835pt}{\includegraphics{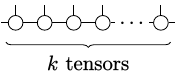}}
    \in \mathcal{V}\otimes \mathcal{H}^{\otimes k} \otimes \mathcal{V}^{*}
  \end{equation}
  is injective.
  The procedure generating tensor \eqref{eq:85} 
  is often referred to as \emph{blocking} \(A\) along \(k\) sites,
  the term \emph{sites} owing to the common application in one-dimensional lattice models.
  If the state generated by blocking \(k\) sites is injective,
  then blocking \(k^{\prime}>k\) sites will also produce an injective
  tensor~\cite{molnarNormalProjectedEntangled2018}.
\end{definition}
\begin{definition}[MPS]
  Every MPS tensor \(A\) generates a family of translationally invariant states
  \(\{\ket{\Psi(A)}_n\mid n\in \mathbb{N}\}\),
  according to
  \begin{equation}
    \label{eq:94}
   \ket{\Psi(A)}_n \coloneq \raisebox{-10.70003pt}{\includegraphics{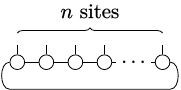}} \in \mathcal{H}^{\otimes n}.
  \end{equation}
  Contracting the two outermost legs is referred to as \emph{periodic boundary conditions}.
  We will not consider other choices of boundary conditions in this paper.
  Choosing a basis for the physical Hilbert space \(\mathcal{H}\),
  the MPS can equivalently be written as
  \begin{equation}
    \label{eq:79}
   \ket{\Psi(A)}_n = \sum_{i_1,\ldots, i_n} \operatorname{tr}\left( A^{i_1}\cdots A^{i_n} \right)\ket{i_1,\ldots, i_n}.
  \end{equation}
  When there is no possibility of confusion,
  we will usually omit the index \(n\) in \(\ket{\Psi(A)}_n\) in the following sections.
\end{definition}
\begin{remark}
  We will often use the notation \(\ket{\Psi(A,B)}\) for the MPS \(\ket{\Psi(AB)}\),
  where \(AB\) denotes the MPS tensor constructed from the tensors \(A\) and \(B\)
  by contracting over \(\mathcal{V}_{\mathrm{t}}^A = \mathcal{V}_{\mathrm{o}}^B\).
\end{remark}
\begin{definition}[Transfer operator]\label{def:transfer-operator}
  Let \(A\) and \(B\) be MPS tensors with the same physical Hilbert space.
  The \emph{(mixed) transfer operator}
  \(\mathbb{T}(A,B)\) of \(A\) and \(B\) is defined to be the tensor
  \begin{equation}
    \label{eq:110}
    \mathbb{T}(A,B)\coloneq\raisebox{-24.2663pt}{\includegraphics{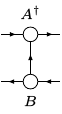}}
    \in \operatorname{End}(\mathcal{V}^{*}_A\otimes \mathcal{V}_B),
  \end{equation}
  where we view the diagram
  (and this is the only exception to the general rule above)
  as a map going from right to left.
  After choosing a basis for the physical Hilbert space, the transfer operator is just the superoperator
  \begin{equation}
    \label{eq:111}
    Y \longmapsto \sum_i B^i Y {A^i}^{\dag}.
  \end{equation}
\end{definition}
The transfer operator is often used to compute the scalar product between translation invariant
matrix product states, which is easily seen to take the form
\begin{equation}
  \label{eq:112}
  \braket{\Psi(A)|\Psi(B)}_n = \operatorname{tr} \mathbb{T}(A,B)^n.
\end{equation}
Crucially, for large \(n\) this expression is dominated by the leading eigenvalues of \(\mathbb{T}\).

Different MPS tensors can give rise to the same MPS,
e.g., the substitution
\begin{equation}
  \label{eq:86}
  A^i\rightarrow X^{-1}A^iX
\end{equation}
will leave the state \(\ket{\Psi(A)}\) invariant,
for any invertible \(d_{\mathcal{V}}\times d_{\mathcal{V}}\) matrix \(X\). 
The transformation \eqref{eq:86} is usually referred to as a \emph{gauge transformation},
and in graphical notation it takes the form
\begin{align*}
  \raisebox{-12.88516pt}{\includegraphics{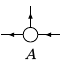}} \; \rightarrow
  \raisebox{-12.88516pt}{\includegraphics{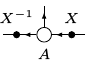}}.
\end{align*}
While unrelated MPS tensors can generally generate the same states,
injective MPS tensors generating the same states are necessarily related via a gauge transformation,
as the following theorem shows.
\begin{theorem-no-proof}[Fundamental theorem of MPSs~\cite{molnarNormalProjectedEntangled2018}]
  \label{thm:fundamental-thm}
  If the injective MPS tensors \(A\) and \(B\) generate the same states
  \begin{equation}
    \label{eq:87}
    \ket{\Psi(A)}_n =
    \ket{\Psi(B)}_n,
  \end{equation}
  for some \(n \geq 3\),
  then \(A\) and \(B\) are related via a gauge transformation.
\end{theorem-no-proof}

\begin{lemma}[Canonical form~\cite{perez-garciaMatrixProductState2007}]
  \label{lem:canonical-form}
  Any MPS \(\ket{\Psi(A)}\) with periodic boundary conditions
  admits a set of normal tensors \(A_1, \ldots, A_m\),
  such that
  \(\ket{\Psi(A)}_n = \ket{\Psi(\widetilde{A})}_n \)
  for all \(n\in \mathbb{N}\), where
  \begin{equation}
    \label{eq:88}
    \widetilde{A}^i \coloneq \bigoplus_{j=1}^m A^i_j
  \end{equation}
  is called the \emph{canonical form} of \(A\).
  (Note that the canonical form, as introduced here, is not unique)
\end{lemma}

\begin{definition}[Global symmetry]
  \label{def:global-symmetry}
  We say that the MPS \(\ket{\Psi(A)}\) has a global on-site symmetry
  generated by the unitary group representation \((\mathcal{H}, U) \in \Rep G\),
  if (for any system size \(n\))
  \begin{equation*}
    U(g)^{\otimes n }\ket{\Psi(A)}_n = \ket{\Psi(A)}_n.
  \end{equation*}
\end{definition}

Matrix product operators (MPOs) can be seen as a generalization of MPSs to describe operators.
While MPSs are used to describe quantum states,
MPOs describe quantum operators, such as Hamiltonians, density matrices or global symmetries.

\begin{definition}[MPO tensor]\label{def:mpo}
  An \emph{MPO tensor} is an order 4 tensor 
  \(T \in \mathcal{V}_{\mathrm{o}}\otimes \mathcal{H}_1 \otimes \mathcal{V}_{\mathrm{t}}^{*} \otimes \mathcal{H}_2^{*}\).
  Again, Hilbert spaces denoted as \(\mathcal{H}\) we shall refer to as \emph{physical},
  while Hilbert spaces denoted as \(\mathcal{V}\) we shall refer to as \emph{virtual}.
  Graphically, we will denote the MPO tensor \(T\) as
  \begin{equation}
    \label{eq:84}
    T \equiv  \raisebox{-11.69995pt}{\includegraphics{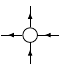}} \equiv  \raisebox{-24.72456pt}{\includegraphics{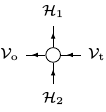}}
  \end{equation}
  and choosing bases for \(\mathcal{H}_1, \mathcal{H}_2, \mathcal{V}_{\mathrm{o}}\) and \(\mathcal{V}_{\mathrm{t}}\),
  the tensor \(T\) can be viewed as a set of matrices
  \(\big\{ T^i_j \mid i = 1, \ldots, d_{\mathcal{H}_1}\text{ and } j=1, \ldots, d_{\mathcal{H}_2} \big\}\)
  of size \(d_{\mathcal{V}_{\mathrm{o}}}\times d_{\mathcal{V}_{\mathrm{t}}}\).
\end{definition}

\begin{definition}[MPO]
  Every MPO tensor \(T\),
  with \(\mathcal{V}_{\mathrm{o}} = \mathcal{V}_{\mathrm{t}}\),
  generates a family of translationally invariant operators
  \(\{\mathcal{Q}_n(T)\mid n\in \mathbb{N}\}\),
  according to
  \begin{equation}
    \label{eq:82}
    \mathcal{Q}_n(T) \coloneq \raisebox{-12.7pt}{\includegraphics{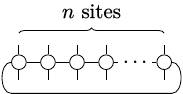}}
    \in \left( \mathcal{H}_1 \otimes \mathcal{H}_2^{*} \right)^{\otimes n}.
  \end{equation}
  Again, contracting the two outermost legs is referred to as \emph{periodic boundary conditions}.
  Choosing bases for the physical Hilbert spaces \(\mathcal{H}_1, \mathcal{H}_2\),
  the MPO can equivalently be written as
  \begin{equation*}
    \mathcal{Q}_n(T) = \sum_{i,j} \operatorname{tr} \left(T^{i_1}_{j_1}\cdots T^{i_n}_{j_n} \right)
    \dyad{i_1, \ldots, i_n}{j_1, \ldots, j_n}.
  \end{equation*}
  As we do with translation-invariant MPSs,
  when there is no possibility of confusion, we will omit the index \(n\) in the following sections.
\end{definition}

\subsection[Spontaneous Symmetry Breaking and Symmetry Protected Topological Phases]{\texorpdfstring{Spontaneous Symmetry Breaking and \\Symmetry Protected Topological Phases}{Spontaneous Symmetry Breaking and Symmetry Protected Topological Phases}}
\label{sec:spt-ssb-phases}

Two gapped, symmetric Hamiltonians are said to be in the same quantum phase
if there exists a smooth path of gapped, symmetric Hamiltonians interpolating between them.

For the case of one-dimensional, \(G\)-symmetric spin systems
a complete classification of the quantum phase space exists
\cite{schuchClassifyingQuantumPhases2011, chenClassificationGappedSymmetric2011, ogataClassificationSymmetryProtected2020, rubioClassifyingSymmetricSymmetrybroken2024},
where distinct phases are uniquely labeled by the tuple of invariants\footnote{A phase invariant is a property or quantity that is shared by all members of the phase.}
\((G_{\mathrm{u}}, [\alpha])\),
with \(G_{\mathrm{u}}\trianglelefteq G\) a normal subgroup 
and \([\alpha] \in \operatorname{H}^2(G_{\mathrm{u}}, \operatorname{U}(1))\)
an equivalence class of 2-cocycles, which is explained in more detail below.
The first (coarser) invariant, referred to as \emph{unbroken} symmetry group,
is related to spontaneous symmetry breaking and gives the ground state degeneracy
through the order \(|G_{\mathrm{b}}|\) of the \emph{broken} symmetry group
\(G_{\mathrm{b}}\coloneq G/G_{\mathrm{u}}\).
The latter invariant discriminates between models of equal ground state degeneracy and
represents the SPT part of the phase.

\begin{example}
  The Ising model can be in the symmetric (disordered) phase \((\mathbb{Z}_2 \leq \mathbb{Z}_2, 1)\),
  or in the symmetry broken (ordered) phase \((1 \leq \mathbb{Z}_2, 1)\),
  featuring unique ground state and two-fold degenerate ground states, respectively.
\end{example}
\begin{example}
  The spin-1 AKLT model belongs to the phase \((\mathbb{Z}_2\times \mathbb{Z}_2, [\alpha]\neq 1)\),
  dubbed Haldane phase.%
\end{example}

Since every injective MPS is the unique ground state of a local, frustration free, gapped Hamiltonian~\cite{ciracMatrixProductStates2021} (dubbed \emph{parent} Hamiltonian),
restricting to Hamiltonians admitting exact MPS representations of their ground states
allows the phase classification to be carried out at the level of ground states~\cite{schuchClassifyingQuantumPhases2011}.
In the following we explain how the invariants introduced above
are encoded in the local degrees of freedom (the generating tensors) of matrix product states.

The MPS tensor generating a state in the phase \((G_{\mathrm{u}}, [\alpha])\)
admits a canonical form, cf.~Lemma~\ref{lem:canonical-form}, that is given by
\begin{equation}\label{eq:91}
  A^i =
  \begin{pmatrix}
    A^i_1&&\\
    &\ddots&\\
    &&A^i_{n_{\mathrm{b}}}
  \end{pmatrix}
  \equiv \bigoplus_{k=1}^{n_{\mathrm{b}}} A_k^i,
\end{equation}
where \(n_{\mathrm{b}} \coloneq |G_{\mathrm{b}}|\),
and each block in the decomposition of \(A\) corresponds to a different non-symmetric ground state.
Since the broken symmetry group permutes between ground states,
in this picture it acts by permuting the normal blocks of \(A\).
We assume that no non-trivial partition of the set of blocks,
invariant under the the full action of \(G\), exists.
Otherwise, the state would exhibit long range entanglement not protected by \(G\),
contradicting robustness, a property usually required of quantum phases,
cf.~\cite{schuchClassifyingQuantumPhases2011}.

Each ground state \(\ket{\Psi(A_k)}, k=1, \dots, n_{\mathrm{b}}\), has a global \(G_{\mathrm{u}}\)-symmetry,
cf.\ Definition~\ref{def:global-symmetry}, hence there exist unitary projective representations
\((\mathcal{V}_k, V_k) \in \Rep_{\mathbb{C}}^{\alpha}G_{\mathrm{u}}\),
cf.\ Definition~\ref{def:projective-reps},
with \([\alpha]\in \mathrm{H}^2(G_{\mathrm{u}}, \mathrm{U}(1))\), such that
(up to a gauge transformation of \(A_k\))
\begin{align}\label{eq:0}
  \raisebox{-13.94629pt}{\includegraphics{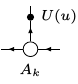}}
  = \raisebox{-13.94629pt}{\includegraphics{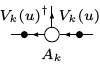}}
\end{align}
for all \(k\) and \(u \in G_{\mathrm{u}}\).
This follows directly from the Fundamental Theorem of MPSs, cf.\ Theorem~\ref{thm:fundamental-thm}.
In other words, we can view the normal blocks of the MPS tensor as intertwiners
\(A_k\in \operatorname{Hom}_{\Rep_{\mathbb{C}}^{\alpha}G_{\mathrm{u}}}(\mathcal{V}_k, \mathcal{V}_k\otimes \mathcal{H})\)
and we shall refer to Equation~\eqref{eq:0} as the \emph{intertwiner property} of \(A_k\).

\begin{definition}[Projective representations]\label{def:projective-reps}
  A \emph{projective representation} of the group \(G\) on the complex vector space \(W\)
  is a collection of linear maps
  \(\{\rho(g)\in \operatorname{GL}(W)\mid g \in G\}\),
  that satisfy
  \begin{equation}
    \label{eq:27}
    \rho(g)\rho(h) = \alpha(g,h) \rho(gh),
  \end{equation}
  with \(\alpha\in \mathrm{Z}^2(G, \mathrm{U}(1))\), the set of complex valued \emph{2-cocycles},
  meaning \(\alpha(g,h) \in \operatorname{U}(1)\subset \mathbb{C}\) and \(\alpha\) satisfies the \emph{2-cocycle condition}
  \begin{equation}
    \label{eq:6}
    \alpha(g, h)\alpha(gh, k) = \alpha(g, hk)\alpha(h,k)
  \end{equation}
  for all \(g,h,k\in G\).
  A projective representation can be thought of as a linear (i.e.\ non-projective) representation
  that satisfies the composition property only up to a scalar.
  Redefining \(\rho(g)\rightarrow \widetilde{\rho}\coloneq c(g)\rho(g)\),
  where \(c(g) \in \mathbb{C}^{\times}\) for all \(g\in G\),
  which implies
  \begin{equation}
    \label{eq:35}
    \alpha(g,h) \rightarrow \widetilde{\alpha}(g,h) \coloneq \frac{c(g)c(h)}{c(gh)}\alpha(g,h),
  \end{equation}
  should give an equivalent projective representation, cf.\ Equation~\eqref{eq:0},
  leading us to identify \(\alpha \sim \widetilde{\alpha}\).
  The set of equivalence classes in \(\mathrm{Z}^2(G, \mathrm{U}(1))\) under this relation
  forms a group, denoted \(\mathrm{H}^2(G, \mathrm{U}(1))\),
  and is usually referred to as the \emph{second cohomology group of \(G\)}.
  We let \([\alpha]\) denote the equivalence class represented by the 2-cocycle \(\alpha\).
  The category of all projective representations of \(G\) with 2-cocycle \(\alpha\) is denoted as
  \(\Rep ^{\alpha}_{\mathbb{C}}G\).
  We have collected a more detailed exposition and all relevant facts concerning projective representations
  of finite groups in Appendix~\ref{sec:proj-repr}.
\end{definition}
\begin{remark}
  Without loss of generality we assume that all 2-cocycles are normalized,
  i.e.\ \(\alpha(g, e) = \alpha(e, g) = 1\) for all \(g\in G\),
  and since we only consider finite groups in this paper we assume that all representations
  and cocycles are unitary.
  We will sometimes refer to the 2-cocycle \(\alpha\) when in fact we mean its equivalence class
  \([\alpha]\in \mathrm{H}^2(G, \mathrm{U}(1))\).
  Furthermore, we will use the term \emph{2-cocycle} also for \([\alpha]\).
  In practice, this abuse of notation should not cause any confusion.
  Similarly, we will often simply use the collection of linear maps \(\rho\)
  to refer to the representation \((W, \rho)\),
  suppressing any explicit mention of the underlying vector space \(W\).
  On the other hand,
  the projective representation above can equivalently be viewed as a
  \(\mathbb{C}^{\alpha}G\)-module \(W\),
  leaving the linear maps \(\rho\) implicit.
  Depending on the context we shall employ the more convenient description.
\end{remark}

\begin{definition}[Slant product]\label{def:slant-product-i}
  We define the \emph{slant product} of \([\alpha]\in \operatorname{H}^2(G, \operatorname{U}(1))\),
  in an abelian group \(G\),
  to be the group homomorphism (cf.\ Appendix~\ref{sec:proj-rep-abelian-groups})
  \begin{align}
    \imath\alpha \colon
    \eqmakebox[i1]{\(G\)} & \longrightarrow \eqmakebox[i2]{\(\widehat{G}\)} \\\nonumber
    \eqmakebox[i1]{\(g\)} & \longmapsto \eqmakebox[i2]{\(\imath_g\alpha\)} \coloneq
    \left( h\mapsto \frac{\alpha(g,h)}{\alpha(h,g)} \right),
  \end{align}
  where \(\widehat{G} \coloneq \operatorname{Hom}(G, \operatorname{U}(1))\)
  is the \emph{Pontryagin dual} of \(G\).
\end{definition}
\begin{definition}[SPT degeneracy]\label{def:degeneracy1}
  Let \(G\) be an abelian group and \([\alpha]\in \mathrm{H}^2(G, \mathrm{U}(1))\).
  We call the subgroup
  \begin{equation}
    \label{eq:12}
    K_{\alpha} \coloneq \left\{ g\in G \mid \forall h\in G \colon \alpha(g,h) = \alpha(h, g)\right\}
  \end{equation}
  the \emph{degeneracy} of \([\alpha]\).
  Note that this group is just the kernel of the slant product of \([\alpha]\),
  cf.\ Definitions~\ref{def:slant-product} and~\ref{def:degeneracy}.
  If \(K_{\alpha}\) is trivial \(\imath\alpha\) is an isomorphism 
  and we call \(\alpha\) and \([\alpha]\) \emph{non-degenerate}.
  The degeneracy of an SPT phase is just the degeneracy of the 2-cocycle characterizing it.
\end{definition}
\begin{definition}[MNC phase]
  An SPT phase,
  protected by an abelian group,
  is called \emph{maximally non-commutative} (MNC)
  if it is characterized by a non-degenerate 2-cocycle.
\end{definition}
Every irreducible projective representation associated to a non-degenerate 2-cocycle is faithful,
hence, the virtual representations of an MPS in an MNC phase have to be faithful.
From a more physical perspective,
MNC phases have maximal inaccessible entanglement, maximal SPT complexity, and
maximal topological edge mode degeneracy \cite{deGrootPhDthesis}.
In practice, these phases are of particular interest because
they can be detected using string order parameters~\cite{pollmannDetectionSymmetryprotectedTopological2012}
and they can act as universal resources for
measurement based quantum computation~\cite{miyakeQuantumComputationEdge2010, elseSymmetryProtectedPhasesMeasurementBased2012, millerResourceQualitySymmetryProtected2015, stephenComputationalPowerSymmetryProtected2017}.

\section{Gauging Procedure}
\label{sec:gauging-procedure}

Gauging is a procedure that transforms globally symmetric quantum states and operators into locally symmetric ones, while preserving their expectation values.
Hence, gauging maps ground states to ground states and preserves the gap of gapped Hamiltonians~\cite{williamsonMatrixProductOperators2016}.
In this section we review parts of the gauging procedure established in \cite{haegemanGaugingQuantumStates2015},
in the context of one-dimensional spin chains,
and introduce a slight generalization thereof,
called \emph{twisted gauging},
which takes as an additional input a 2-cocycle
\(\tau\in \operatorname{Z}^2(G, \operatorname{U}(1))\)
and reduces to the original untwisted procedure for trivial \emph{twist} \(\tau\).
We emphasize that the twisted gauging procedure includes the untwisted one as a special case,
which is why we shall state general results (whenever possible) in the context of twisted gauging.
We review how to formulate the gauging procedure as the application of an MPO
and derive a convenient MPS representation of the state produced by gauging an MPS.
Both the twisted and untwisted procedures aim to produce a locally symmetric state,
from a given globally symmetric one,
by adding auxiliary degrees of freedom in some fixed initial configuration
and then symmetrizing the product state,
where the untwisted version uses linear representations for the new degrees of freedom,
while the twisted version generally uses projective representations.
In both cases the gauged state will exhibit a global symmetry,
dubbed the \emph{dual} symmetry,
in addition to the local symmetry it has by construction,
as we will see below.
In this paper we focus solely on the gauging of states
and do not discuss the definition of gauged operators or subtleties arising therein. 
\subsection{Untwisted Gauging}
\label{sec:untwisted-gauging}
Let \(G\) be a finite group and
consider a globally \(G\)-symmetric state \(\ket{\psi}\) on a one-dimensional lattice.
The lattice can be infinite,
in which case we label the lattice sites by integers
\(i \in \Lambda \coloneq \mathbb{Z}\),
or it can be of finite length \(n \in \mathbb{N}\),
with periodic boundary conditions,
in which case we label the lattice sites by integers modulo \(n\), i.e.\
\(i \in \Lambda \coloneq \mathbb{Z}_n \equiv \mathbb{Z}/n\mathbb{Z}\).
To each lattice site \(i\) we will associate local spin degrees of freedom,
described by the Hilbert space \(\mathcal{H}_i\cong \mathbb{C}^{d_i}\).
We will collectively refer to these spaces
as \emph{matter} degrees of freedom.
We demand that
\((\mathcal{H}_i,U) \in \Rep G \equiv \Rep_{\mathbb{C}}G\),
i.e.\ the Hilbert space \(\mathcal{H}_i\) comes with a
linear representation
\(U_i\colon G\rightarrow \operatorname{GL}(\mathcal{H}_i)\)
of the group \(G\),
which we assume to be unitary.
The state being globally symmetric means it is invariant under
\begin{equation}
  \label{eq:13}
  U_{\Lambda}(g) \coloneq \bigotimes_{i\in \Lambda} U_i(g),
\end{equation}
i.e.\ \(U_{\Lambda}(g)\ket{\psi} = \ket{\psi}\) for all \(g \in G\).
In the case where \(U_i = U\) for all \(i\in \Lambda\),
this reduces to \(U_{\Lambda} = U^{\otimes |\Lambda |}\), cf.\ Definition~\ref{def:global-symmetry}.
To produce a locally symmetric state, we begin by introducing additional
(so called \emph{gauge}) degrees of freedom,
which we think of as being associated to the edges of the lattice,
i.e.\ sitting in between the matter degrees of freedom, cf.\ Figure~\ref{fig:gauge-dof}.

\begin{figure}[H]
  \centering\raisebox{-0pt}{\includegraphics{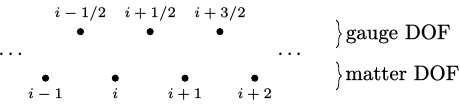}}
  \caption{
    We imagine the original degrees of freedom, now referred to as \emph{matter},
    as sitting on the vertices of a (periodic or infinite) 1D lattice.
    The new degrees of freedom, referred to as \emph{gauge},
    are then each placed between two matter degrees of freedom, i.e.\ on the edges of the lattice.
  }
  \label{fig:gauge-dof}
\end{figure}
Accordingly, we label the corresponding local Hilbert spaces \(\mathcal{H}_e\)
by half-integers \(e \in \Lambda+\frac{1}{2}\).
On the new degrees of freedom we have to choose representations 
\((\mathcal{H}_e ,L_e), (\mathcal{H}_e, R_e) \in \Rep G\),
which we require to commute,
i.e.\ their operators have to satisfy 
\begin{equation}\label{eq:43}
  L_e(g)R_e(h) = R_e(h)L_e(g)
\end{equation}
for all \(g, h\in G\) and \(e\in \Lambda+\frac{1}{2}\).
The prototypical example, which inspires the notation and on which we will focus in the sections below,
is the left and right regular representations on \(\mathbb{C}G\),
but for now we will continue to consider the general case.
The local \(G\)-symmetry, i.e. the gauge symmetry,
shall be implemented by the three-site operators
\(R_{i-\frac{1}{2}}(g) \otimes U_i(g) \otimes L_{i+\frac{1}{2}}(g)\),
so for each \(i\in \Lambda\) we define a projector \(P_i\),
acting on
\(\mathcal{H}_{i-\frac{1}{2}}\otimes \mathcal{H}_i\otimes \mathcal{H}_{i+\frac{1}{2}}\),
given by
\begin{equation}
  \label{eq:16}
  P_i \coloneq \frac{1}{|G|}\sum_{g\in G}
  R_{i-\frac{1}{2}}(g) \otimes U_i(g) \otimes L_{i+\frac{1}{2}}(g),
\end{equation}
which can be interpreted as enforcing a local Gauss law on the lattice site \(i\).
Notice that \(P_i^2 = P_i\) and Equation~\eqref{eq:43} implies
\([P_i,P_j] = 0\) for all \(i,j\in \Lambda\).
The projector to the locally symmetric subspace is now simply given as
\begin{equation}
  \label{eq:19}
  P \coloneq \bigotimes_{i\in \Lambda} P_i.
\end{equation}
We justify this statement by noting that \(P\) is left 
absorbing with respect to the local representations
\(R_{i-\frac{1}{2}}\otimes U_i \otimes L_{i+\frac{1}{2}}\), i.e.
\begin{align}
  P &= \left( R_{i-\frac{1}{2}}(g)\otimes U_i(g)\otimes L_{i+\frac{1}{2}}(g) \right) \circ P
\end{align}
for all \(g\in G\) and \(i \in \Lambda\).
Crucially, a state in the image of \(P\) will be symmetric
with respect to all representations 
\(R_{i-\frac{1}{2}}\otimes U_i \otimes L_{i+\frac{1}{2}}\).
In particular, given an arbitrary auxiliary state
\begin{equation}
  \ket{\Omega} \in \bigotimes_{i\in \Lambda} \mathcal{H}_{i+\frac{1}{2}},
\end{equation}
the state \(P(\ket{\psi}\otimes \ket{\Omega})\) will be gauge invariant.
A common choice, which we shall also make below,
is to use a product auxiliary state.
\begin{definition}[Untwisted Gauging]\label{def:untwisted-gauging}
  We define \emph{untwisted gauging}, or simply \emph{gauging}, to be the map
  \begin{equation}
    \label{eq:52}
    \mathcal{G}\colon \ket{\psi} \longmapsto P(\ket{\psi}\otimes \ket{\Omega}),
  \end{equation}
  where \(P\) denotes the projector introduced above, 
  for the specific choice of gauge degrees of freedom given by
  \begin{equation}
    (\mathcal{H}_e, L_e, R_e) \coloneq (\mathbb{C}G, L, R)
  \end{equation}
  for all \(e\in \Lambda+\frac{1}{2}\),
  where \(\mathbb{C}G\) denotes the group algebra of \(G\) and
  \(L\) and \(R\) are the left and right regular representations of \(G\),
  cf.~Definition~\ref{def:regular-projective-main} below for \(\alpha = 1\),
  and auxiliary state
  \begin{equation}
    \ket{\Omega}
    \coloneq \bigotimes_{i\in \Lambda} \ket{1}_{i+\frac{1}{2}}
    \equiv \bigotimes_{i\in \Lambda} \ket{1},
  \end{equation}
  where \(1\in G\) denotes the unit element of the group.
\end{definition}

\subsection{Twisted Gauging}
\label{sec:twisted-gauging}
We can generalize the gauging procedure by equipping
the gauge degrees of freedom with projective representations
\((\mathcal{H}_e,\tw{\tau}{R}_e)\in \Rep^{\tau}_{\mathbb{C}}G\)
and
\((\mathcal{H}_e,\tw{\overline{\tau}}{L}_e)\in \Rep^{\overline{\tau}}_{\mathbb{C}}G\),
where \([\tau]\in \mathrm{H}^2(G, \operatorname{U}(1))\),
instead of the linear representations above.
Since \([\overline{\tau}] = [\tau]^{-1}\), the representation
\(\tw{\tau}{R}_{i-\frac{1}{2}} \otimes U_i \otimes \tw{\overline{\tau}}{L}_{i+\frac{1}{2}}\),
which defines the gauge symmetry,
remains linear also in the twisted case.
We extend the notation for the projections \eqref{eq:16} and \eqref{eq:19},
to reflect the presence of a twist:
\begin{equation}
  \tw{\tau}{P}_i \coloneq \frac{1}{|G|}\sum_{g\in G}
  \tw{\tau}{R}_{i-\frac{1}{2}}(g) \otimes U_i(g) \otimes \tw{\overline{\tau}}{L}_{i+\frac{1}{2}}(g),
  \tag*{\(\tw{\tau}{(\ref{eq:16})}\)}
\end{equation}
\begin{equation}\label{eq:46}
  \tw{\tau}{P} \coloneq \bigotimes_{i\in \Lambda} \tw{\tau}{P}_i.
  \tag*{\(\tw{\tau}{(\ref{eq:19})}\)}
\end{equation}
Everything discussed in the previous section continues to hold also in the twisted version.
Of course, for trivial \([\tau]\) the twisted gauging reduces to the untwisted procedure,
as can be checked easily.
Hence, we shall (whenever possible) only consider the twisted case going forward.
\begin{remark}
  We use the notation \(\tw{\tau}{(-)}\) to indicate the twisted version of \((-)\).
  In particular, \(\tw{\tau}{(-)}\) always reduces to \((-)\) for trivial \([\tau]\).
\end{remark}
\begin{definition}[Twisted Gauging]\label{def:twisted-gauging}
  Given an element of the second cohomology group \([\tau]\in \operatorname{H}^2(G, \operatorname{U}(1))\)
  we define \emph{twisted gauging},
  or \(\tw{\tau}{\text{gauging}}\), to be the map
  \begin{equation}
    \label{eq:55}
    \tw{\tau}{\mathcal{G}}\colon \ket{\psi} \longmapsto \tw{\tau}{P}(\ket{\psi}\otimes \ket{\Omega}),
  \end{equation}
  described above, 
  for the specific choice of gauge degrees of freedom given by
  \begin{equation}
    (\mathcal{H}_e, \tw{\overline{\tau}}L_e, \tw{\tau}R_e)
    \coloneq (\mathbb{C}G, \tw{\overline{\tau}}L, \tw{\tau}R),
  \end{equation}
  for all \(e\in \Lambda+\frac{1}{2}\),
  where \(\tw{\overline{\tau}}{L}\) and \(\tw{\tau}{R}\)
  are regular projective representations of \(G\) (defined below) and auxiliary state
  \begin{equation}
    \ket{\Omega}
    \coloneq \bigotimes_{i\in \Lambda} \ket{1}_{i+\frac{1}{2}}
    \equiv \bigotimes_{i\in \Lambda} \ket{1}.
  \end{equation}
  For trivial \([\tau]\) this reduces to Definition~\ref{def:untwisted-gauging}, as required.
\end{definition}
\begin{definition}[Regular projective representations]\label{def:regular-projective-main}
  For \(\alpha\in \mathrm{Z}^2(G, \mathrm{U}(1))\)
  one defines on \(\mathbb{C}G\) the \(\alpha\)-projective representations
  \begin{align}
    \tw{\alpha}{L}(g) &:= \sum_{h\in G} \alpha(g, h) \dyad{gh}{h},\\
    \tw{\alpha}{R}(g) &:= \sum_{h\in G} \alpha(h, g) \dyad{h}{hg},
  \end{align}
  referred to as \emph{left} and \emph{right regular \(\alpha\)-projective representations},
  respectively.
  \(\tw{\alpha}{L}\) commutes with \(\tw{\overline{\alpha}}{R}\).
  For more details we refer the reader to Appendix~\ref{sec:proj-repr}.
\end{definition}
\subsection{Gauging as an MPO}
\label{sec:gauging-MPO}

The gauging procedure just described turns out to have a straightforward interpretation
as a matrix product operator,
the explicit construction of which is most naturally carried out using graphical notation.
For the sake of brevity, we only construct the MPO representation
of the specific gauging map given in Definition~\ref{def:twisted-gauging},
at the risk of obscuring some insights regarding the general structure of the constituent tensors.
For instance, the virtual bonds of the MPO will be given by the group algebra,
irrespective of the choice of \(\mathcal{H}_e\).

Let us denote by \(\ket{g}\) the canonical basis vectors of
the group algebra \(\mathbb{C}G\),
where in tensor network diagrams we use the depictions
\({\raisebox{-2.5pt}{\includegraphics{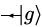}}\colon \mathbb{C}\rightarrow\mathbb{C}G}\) and
\({\raisebox{-2.5pt}{\includegraphics{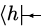}}\colon \mathbb{C}G\rightarrow\mathbb{C}}\)
in order to emphasize domain (incoming leg) and codomain (outgoing leg).
The main building blocks of the gauging procedure, and therefore its MPO description,
are the projectors \(\tw{\tau}{P}_i\), cf.\ Equation~\ref{eq:46},
which we can depict graphically as
\begin{equation}
  \tw{\tau}{P}_i = \frac{1}{|G|} \raisebox{-11.92636pt}{\includegraphics{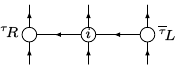}},
\end{equation}
where
\begin{subequations}
  \label{eq:31}
  \begin{align}
    \eqmakebox[eq31]{\(\raisebox{-11.92636pt}{\includegraphics{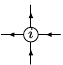}}\)}
    &\coloneq \sum_{g\in G} \raisebox{-2.5pt}{\includegraphics{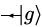}}\; \raisebox{-11.92636pt}{\includegraphics{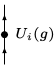}}
      \raisebox{-2.5pt}{\includegraphics{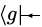}},\\
    \label{eq:32}
    \eqmakebox[eq31][l]{\(\raisebox{-11.92636pt}{\includegraphics{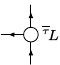}}\)}
    &\coloneq \sum_{g\in G} \raisebox{-2.5pt}{\includegraphics{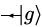}}\; \raisebox{-11.92636pt}{\includegraphics{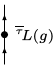}},\\
    \label{eq:33}
    \eqmakebox[eq31][r]{\(\raisebox{-11.92636pt}{\includegraphics{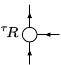}}\)}
    &\coloneq \sum_{g\in G} \raisebox{-11.92636pt}{\includegraphics{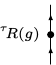}}\;
      \raisebox{-2.5pt}{\includegraphics{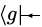}}.
  \end{align}
\end{subequations}
\begin{lemma}[Gauging MPO]\label{lem:gauging-mpo}
  Twisted gauging,
  as it is defined in Definition~\ref{def:twisted-gauging},
  admits an MPO representation 
  \begin{equation*}
    \tw{\tau}{\mathcal{G}} = \raisebox{-18.75801pt}{\includegraphics{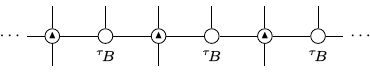}},
  \end{equation*}
  where we use {\small\(\blacktriangle\)} as a placeholder for the appropriate vertex label and
  \begin{equation}\label{eq:92}
    \tw{\tau}{B} \coloneq
    \frac{1}{|G|} \sum_{g \in G}
    \Stk{\raisebox{-10.61665pt}{\includegraphics{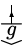}}}
        {\raisebox{-13.31592pt}{\includegraphics{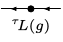}}}.
  \end{equation}
\end{lemma}
\begin{proof}
  If we can show that
  \begin{equation} \label{eq:18}
    \raisebox{-15.27237pt}{\includegraphics{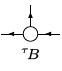}}
    = \frac{1}{|G|}\, \raisebox{-23.8988pt}{\includegraphics{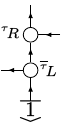}}
    \equiv \frac{1}{|G|}\,\raisebox{-23.8988pt}{\includegraphics{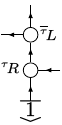}}\,,
  \end{equation}
  then it is clear that the above MPO corresponds to the map 
  \begin{equation}
    \ket{\psi}\mapsto\bigotimes_{i\in\Lambda}\tw{\tau}{P}_i\left(\ket{\psi}\otimes\ket{\Omega}\right).
  \end{equation}
  The equivalence depicted in Equation~\eqref{eq:18}
  follows directly from the commutativity of the involved regular representations and inserting
  \begin{equation}
    \tw{\overline{\tau}}{L}(g) \tw{\tau}{R}(h) \ket{1}
    = \frac{\tau(h^{-1}, h)}{\tau(g, h^{-1})} \ket{gh^{-1}}
  \end{equation}
  yields
  \begin{align*}
    \raisebox{-23.8988pt}{\includegraphics{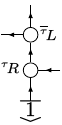}}\,
    &= \sum_{g,h\in G}
    \frac{\tau(h^{-1}, h)}{\tau(g, h^{-1})}
    \Stk{\raisebox{-18.61769pt}{\includegraphics{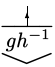}}}
        {\raisebox{-2.5pt}{\includegraphics{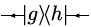}}}\\
    &= \sum_{g,h\in G}
    \tau(g,h)
    \Stk{\raisebox{-10.61665pt}{\includegraphics{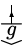}}}
        {\raisebox{-2.5pt}{\includegraphics{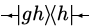}}}
    = \sum_{g \in G}
    \Stk{\raisebox{-10.61665pt}{\includegraphics{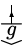}}}
        {\raisebox{-13.31592pt}{\includegraphics{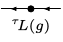}}},
  \end{align*}
  which is what we had to show. Going from the first to the second line we used the
  cocyle condition of \(\tau\).
\end{proof}
Alternatively, we could use Lemma~\ref{lem:gauging-mpo} to define twisted gauging
as the application of the MPO constructed there.
This is convenient because the defining properties of the gauging map can easily be derived
from the intertwiner properties of the constituent tensors.
We can check by explicit computation that the tensors defined in \eqref{eq:31} satisfy
the intertwiner relations
\begin{align*}
  \raisebox{-14.7716pt}{\includegraphics{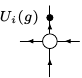}} = 
  \raisebox{-11.92636pt}{\includegraphics{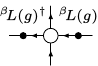}}, &&
  \raisebox{-15.71405pt}{\includegraphics{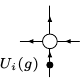}} = 
  \raisebox{-11.92636pt}{\includegraphics{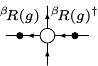}}
\end{align*}
for any \(\beta \in \operatorname{H}^2(G, \operatorname{U}(1))\), and
\begin{align*}
  \raisebox{-14.7716pt}{\includegraphics{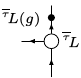}} =
  \raisebox{-11.92636pt}{\includegraphics{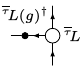}},&&
  \raisebox{-14.7716pt}{\includegraphics{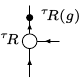}} =
  \raisebox{-11.92636pt}{\includegraphics{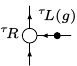}},
\end{align*}
the latter of which imply in particular that
\begin{align*}
  \raisebox{-15.27237pt}{\includegraphics{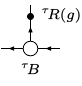}} =
  \raisebox{-15.27237pt}{\includegraphics{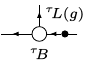}}, &&
  \raisebox{-15.27237pt}{\includegraphics{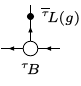}} =
  \raisebox{-15.27237pt}{\includegraphics{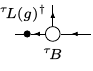}}.
\end{align*}
\subsection{Emergent Dual Symmetry}
\label{sec:repg-sym}
It is well known~\cite{haegemanGaugingQuantumStates2015, thorngrenFusionCategorySymmetry2024, lootensDualitiesOneDimensionalQuantum2023, garre-rubioEmergent2+1DTopological2024}
that the gauged state has,
in addition to the local symmetry arising by construction,
a global \(\Rep G\) symmetry, i.e.\ a categorical symmetry, which is said to be \emph{emergent}.
For any \(\sigma \in \Rep G\) consider the following MPO tensor
\begin{equation}
  \label{eq:75}
  \raisebox{-11.92636pt}{\includegraphics{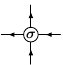}}
  \coloneq \sum_{g\in G}
    \mStk[2.1em]{{\raisebox{-10.61665pt}{\includegraphics{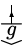}}}
               {\raisebox{-13.06595pt}{\includegraphics{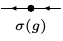}}}
               {\raisebox{-11.72636pt}{\includegraphics{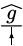}}}},
\end{equation}
the application of which to \(\tw{\tau}{B}\) yields
\begin{equation}
  \raisebox{-25.27237pt}{\includegraphics{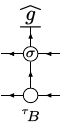}}=
  \Stk{\raisebox{-13.06595pt}{\includegraphics{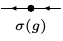}}}
      {\raisebox{-13.31592pt}{\includegraphics{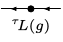}}} \cong
  \Stk{\raisebox{-13.31592pt}{\includegraphics{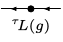}}}
      {\raisebox{-0pt}{\includegraphics{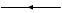}}} = 
  \raisebox{-8.625pt}{\includegraphics{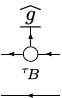}}
\end{equation}
for all \(g\in G\),
where the isomorphism is given by conjugation with
\(\sum_h\sigma(h)\otimes\dyad{h}{h}\) and an exchange of legs.
Note that in diagrams such as above, we order tensor product factors from top to bottom.
We recognize that the corresponding (properly normalized) MPO
\begin{equation}
  \label{eq:68}
  \mathcal{Q}_{\sigma} \coloneq \frac{1}{d_{\sigma}}\cdots \raisebox{-11.67966pt}{\includegraphics{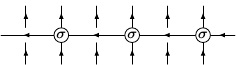}} \cdots
\end{equation}
leaves any gauged state invariant.
In terms of the gauging MPO we have
\(\mathcal{Q}_{\sigma}\circ \tw{\tau}{\mathcal{G}} = \tw{\tau}{\mathcal{G}}\)
and in the language of
\cite{lootensDualitiesOneDimensionalQuantum2023,lootensDualitiesOnedimensionalQuantum2024}
the gauging operator implements a duality between 
\(\cat{Vec}_G\)- and \(\Rep G\)-symmetric systems.

From Equation~\eqref{eq:75} it is easy to see that \(\mathcal{Q}\) constitutes a representation of \(\Rep G\),
since
\begin{align*}
  \raisebox{-21.92636pt}{\includegraphics{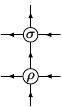}}
  = \sum_{g\in G}
  \mStk[2em]{{\raisebox{-10.61665pt}{\includegraphics{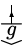}}}
             {\raisebox{-13.06595pt}{\includegraphics{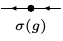}}}
             {\raisebox{-13.06595pt}{\includegraphics{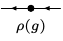}}}
             {\raisebox{-11.72636pt}{\includegraphics{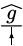}}}}
  =\raisebox{-11.92636pt}{\includegraphics{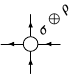}} \hspace{-.6em} \cong
  \bigoplus_{\pi} m_{\pi}^{\sigma\otimes \rho}\raisebox{-11.92636pt}{\includegraphics{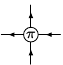}},
\end{align*}
where the sum goes over all irreducible \(\pi\in \Rep G\) and \(m^{\sigma\otimes \rho}_{\pi}\)
denotes the multiplicity of \(\pi\) in \(\sigma\otimes \rho\).
Hence,
\begin{align}
  \label{eq:73}
  \mathcal{Q}_{\sigma_1} \circ \mathcal{Q}_{\sigma_2}
  = \mathcal{Q}_{\sigma_1\otimes \sigma_2}
  = \sum_{\sigma} m^{\sigma_1\otimes \sigma_2}_{\sigma} \mathcal{Q}_{\sigma}.
\end{align}

Notice that the choice of dual symmetry, i.e.\ Equation~\eqref{eq:68}, is not canonical,
since for any isomorphism \(f\colon G\xrightarrow{\sim} H\) of groups we can construct
\begin{equation}
  \label{eq:70}
  \widehat{\mathcal{Q}}_{s} \coloneq \mathcal{Q}_{f^{\star}s},
\end{equation}
exhibiting the gauged MPS as a \(\Rep H\)--symmetric state.
While this observation might appear trivial at the moment,
it will become relevant when we analyze the mapping of phases under gauging in Section~\ref{sec:mapping-phases}.
For now, let us agree to fix the dual symmetry to be given by Equation~\eqref{eq:68}
and revisit this issue later.
\subsection{Gauged MPS}
\label{sec:gauged-mps}
Since the gauging operator can be written as an MPO, it is clear, cf.\ Figure \ref{fig:gauged-mps},
that gauging an MPS will yield again an MPS.
In this section we derive a convenient choice of generating tensors for this \emph{gauged MPS}.
\begin{figure}[H]
  \centering
  \raisebox{-44.8734pt}{\includegraphics{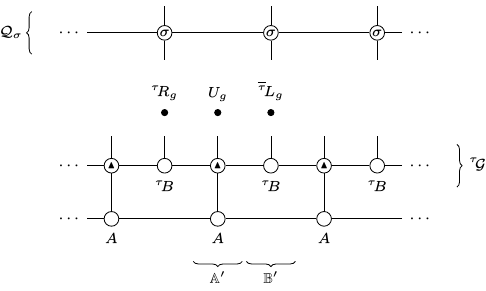}}
  \caption{
    Applying the gauging MPO from Lemma~\ref{lem:gauging-mpo} to a translation invariant MPS
    yields a two-site translation invariant MPS generated by the tensors
    \(\mathbb{A}^{\prime}\) and \(\mathbb{B}^{\prime}\) as indicated in the figure.
    We also depict schematically on which degrees of freedom the local symmetry
    \(\tw{\tau}{R}\otimes U\otimes \tw{\overline{\tau}}{L}\)
    and the emergent dual symmetry \(\mathcal{Q}\) act.
    Generalizing the figure to the non--translation invariant case is straightforward.
  }
  \label{fig:gauged-mps}
\end{figure}
\begin{lemma}\label{lem:gauged-mps}
  Let \(\ket{\Psi(A)}\) be a globally symmetric MPS, generated by the tensor \(A\),
  satisfying the intertwiner relation,
  \begin{align}\label{eq:49}
    \raisebox{-12.88516pt}{\includegraphics{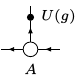}}
    = \raisebox{-12.88516pt}{\includegraphics{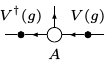}}
  \end{align}
  for all \(g \in G\), where \(V \in \Rep^{\alpha}_{\mathbb{C}} G\).
  Applying the gauging MPO of Lemma~\ref{lem:gauging-mpo} yields an MPS generated by the tensors
  \begin{align*}
    \raisebox{-15.03618pt}{\includegraphics{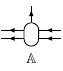}} 
    \coloneq \mStk[2em]{{} {\raisebox{-12.88516pt}{\includegraphics{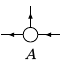}}} {\raisebox{-0pt}{\includegraphics{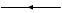}}}},&&
    \raisebox{-17.3845pt}{\includegraphics{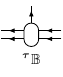}}
    \coloneq \frac{1}{|G|} \sum_{g\in G}
    \mStk[2em]{{\raisebox{-10.61665pt}{\includegraphics{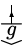}}}
               {\raisebox{-13.06595pt}{\includegraphics{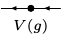}}}
               {\raisebox{-15.68393pt}{\includegraphics{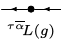}}}}.
  \end{align*}
  Put succinctly, \(\tw{\tau}{\mathcal{G}} \ket{\Psi(A)} = \ket{\Psi(\mathbb{A}, \twb{\tau}{\mathbb{B}})}\).
\end{lemma}
\begin{proof}
  It is clear, e.g., from Figure~\ref{fig:gauged-mps}, that
  \(\tw{\tau}{\mathcal{G}} \ket{\Psi(A)} = \ket{\Psi(\mathbb{A}^{\prime}, \mathbb{B}^{\prime})}
  = \ket{\Psi(\mathbb{A}^{\prime\prime}, \mathbb{B}^{\prime\prime})}\),
  where
  \begin{equation}
    \raisebox{-16.00954pt}{\includegraphics{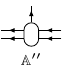}}
    \coloneq \sum_{g\in G}
    \Stk{\raisebox{-12.88516pt}{\includegraphics{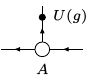}}}
        {\raisebox{-2.5pt}{\includegraphics{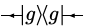}}\hfill}
    = \sum_{g\in G}
    \Stk{\hspace{-.2em}\raisebox{-12.88516pt}{\includegraphics{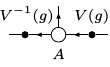}}}
        {\raisebox{-2.5pt}{\includegraphics{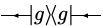}}},
  \end{equation}
  \begin{equation}
    \raisebox{-16.00954pt}{\includegraphics{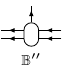}}
    \coloneq \sum_{g\in G}
    \mStk[2em]{{\raisebox{-10.61665pt}{\includegraphics{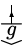}}}
               {\raisebox{-0pt}{\includegraphics{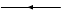}}}
               {\raisebox{-0pt}{\includegraphics{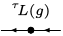}}}}.
  \end{equation}
  We have introduced \(\mathbb{A}^{\prime\prime}\) and \(\mathbb{B}^{\prime\prime}\),
  which are just the tensors \(\mathbb{A}^{\prime}\) and \(\mathbb{B}^{\prime}\), respectively,
  with the order of their virtual legs exchanged,
  to make the tensors easier to draw.
  On the virtual Hilbert space of \(\mathbb{A}^{\prime\prime}\) and \(\mathbb{B}^{\prime\prime}\)
  we define the basis transformation
  \begin{equation}
    S \coloneq \sum_{g\in G} V(g) \otimes \dyad{g}{g}.
  \end{equation}
  Then one finds (see Appendix~\ref{sec:basis-transf} for the explicit computations) that
  \begin{subequations}
   \label{eq:3}
    \begin{align}\label{eq:3a}
      (\mathbb{A}^{\prime\prime})^{\mu}
      &= S^{-1}\big(\eqmakebox[S]{\(A^{\mu} \otimes \mathbbm{1}_n\)}\big)S,\\\label{eq:3b}
      (\mathbb{B}^{\prime\prime})^g
      &= S^{-1}\big(\eqmakebox[S]{\,\(V(g)\otimes \tw{\tau\overline{\alpha}}{L}(g)\,\)}\big)S,
    \end{align}
  \end{subequations}
  hence, the tensors \(\{ \mathbb{A}, \twb{\tau}{\mathbb{B}} \}\)
  generate the same MPS as the tensors
  \(\left\{ \mathbb{A}^{\prime\prime}, \mathbb{B}^{\prime\prime} \right\}\).
\end{proof}
Notice that the tensor \(\tw{\tau}{B}\), cf.\ Equation~\eqref{eq:92},
has a virtual symmetry, given by the representation \(\tw{\overline{\tau}}{R}\).
Hence, the gauged MPS has a virtual symmetry, given by \(\tw{\overline{\tau}}{R}(g) \otimes V(g)\),
which, after exchanging the order of the Hilbert spaces, transforms to
\begin{equation}
  V(g)\otimes\tw{\overline{\tau}}{R}(g) 
  = S^{-1}\left(\eqmakebox[basis]{\(\mathbbm{1} \otimes \tw{\alpha\overline{\tau}}{R}(g) \)}\right)S.
\end{equation}
\begin{remark}
  The only reason why we restrict to translation invariant states in Lemma~\ref{lem:gauged-mps},
  is to simplify notation.
  Generally, if the state \(\ket{\Psi(A)}\) is generated by the tensors \(A_i\), \(i \in \Lambda\),
  satisfying the intertwiner relations,
  \begin{align}\label{eq:56}
    \raisebox{-13.88516pt}{\includegraphics{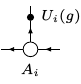}}
    = \raisebox{-13.88516pt}{\includegraphics{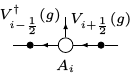}}
  \end{align}
  for all \(g \in G\), where \(V_e \in \Rep^{\alpha}_{\mathbb{C}} G\), for all \(e \in \Lambda+\frac{1}{2}\),
  the gauged MPS is then generated by the tensors
  \begin{align*}
    \raisebox{-16.03618pt}{\includegraphics{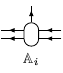}} 
    \coloneq \mStk[2em]{{} {\raisebox{-13.88516pt}{\includegraphics{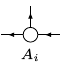}}} {\raisebox{-0pt}{\includegraphics{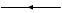}}}},&&
    \raisebox{-17.3845pt}{\includegraphics{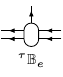}}
    \coloneq \frac{1}{|G|} \sum_{g\in G}
    \mStk[2em]{{\raisebox{-10.61665pt}{\includegraphics{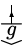}}}
               {\raisebox{-13.06595pt}{\includegraphics{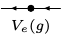}}}
               {\raisebox{-15.68393pt}{\includegraphics{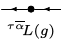}}}}.
  \end{align*}
  The proof proceeds analogously to the one above.
  Also, everything discussed in the remaining sections of the paper continues to hold
  in the non--translation invariant case.
\end{remark}
\subsection{GHZ Structure After Gauging}
\label{sec:ghz-structure}
Conventional wisdom \cite{haegemanGaugingQuantumStates2015} says that
the application of the non-twisted gauging procedure described above
to an injective MPS results in a non-injective state.
However, we will see that the structure of the gauged state
depends intimately on the phase of the initial state.
For instance, gauging an MPS that is initially in an MNC phase
will result in a gauged MPS that is also injective, as we show below.
In this section we analyze the case of a general finite group \(G\).
Once we restrict our attention to abelian groups,
we shall repeat some of the arguments here in more detail.
Hence, this section should also serve as less technical introduction to the themes of later sections.

Consider for now a translation invariant, injective MPS \(\ket{\Psi(A)}\) in the phase given by
\([\alpha]\in \operatorname{H}^2(G, \operatorname{U}(1))\)
and denote by \(\rho_1, \ldots, \rho_r\) the distinct, i.e.\ pairwise linearly inequivalent,
irreducible projective representations of \(G\) with 2-cocycle \(\overline{\alpha}\).
We let \(d_j \coloneq d_{\rho_j} \equiv |\rho_j|\) denote the degree (or dimension) of \(\rho_j\).
Just as in the linear case, the left regular \(\overline{\alpha}\)-projective representation
\(\tw{\overline{\alpha}}{L}\), cf.\ Definition~\ref{def:regular-projective-main}, decomposes as
\begin{equation}\label{eq:2}
  \tw{\overline{\alpha}}{L} \cong
  \bigoplus_{j=1}^r \mathbbm{1}_{d_j}\otimes \rho_j,
\end{equation}
which we show in Corollary~\ref{cor:regular-decomp}.
Notice that the basis transformation
implementing Equation~\eqref{eq:2} for \(\mathbb{B}\)
commutes with all \(\mathbb{A}^{\mu}\),
hence it block diagonalizes the tensors generating the gauged MPS
\(\ket{\Psi(\mathbb{A}, \mathbb{B})}\).
Furthermore, since Equation~\eqref{eq:2} implies \(|G| = d_1^2 + \cdots + d_r^2\),
we can find a common basis transformation such that 
\begin{subequations}\label{eq:1}
  \begin{align}
    \mathbb{A}^{\mu} &\cong \bigoplus_{j=1}^r \mathbbm{1}_{d_j}\otimes
      \underbrace{\left(\mathbbm{1}_{d_j}\otimes A^{\mu}\right)}_{=\mathbb{A}^{\mu}_{\rho_j}},\\
    \mathbb{B}^g &\cong \bigoplus_{j=1}^r \mathbbm{1}_{d_j}\otimes
      \underbrace{\left(\rho_j(g)\otimes V(g)\right)}_{=\mathbb{B}^g_{\rho_j}},
  \end{align}
\end{subequations}
where, here and below, we use \(\cong\) to indicate equality up to conjugation
by the same invertible matrix
and we introduced the notation
\begin{align}\label{eq:4}
  \mathbb{A}^{\mu}_{\lambda} \coloneq \mathbbm{1}_{|\lambda|}\otimes A^{\mu},&&
  \mathbb{B}^g_{\lambda} \coloneq \lambda(g)\otimes V(g),
\end{align}
for any \(\lambda\in\Rep_{\mathbb{C}}^{\overline{\alpha}}G\).
If \(A\) is injective (as we have assumed here) and \(\lambda\) is irreducible then the two site tensors
\(\mathbb{A}_{\lambda}\mathbb{B}_{\lambda}\) and \(\mathbb{B}_{\lambda}\mathbb{A}_{\lambda}\)
are also injective,
hence the decomposition \eqref{eq:1} represents a decomposition of the MPS tensors into injective blocks.
The case when \(G\) has only one projective irrep with 2-cocycle \(\alpha\)
is particularly simple to analyze.
\begin{example}[MNC phases]\label{ex:mnc-phases}
  If \(G\) is abelian and \(\alpha\) is MNC,
  there exists only a single irreducible \(\overline{\alpha}\)-projective representation
  up to linear equivalence and its degree is given by \(d\coloneq \sqrt{|G|}\).
  In this case Equation~\eqref{eq:2} reduces to
  \(\tw{\overline{\alpha}}{L} \cong \mathbbm{1}_d\otimes \rho\)
  and Equations \eqref{eq:1} reduce to
  \begin{align}
    \mathbb{A}^{\mu} \cong \mathbbm{1}_d\otimes \mathbb{A}^{\mu}_{\rho},&&
    \mathbb{B}^g \cong \mathbbm{1}_d\otimes \mathbb{B}^g_{\rho},
  \end{align}
  implying that \(\mathbb{A}\) and \(\mathbb{B}\) each consists of \(m\) identical blocks.
  We can discard the duplicate blocks (they only contribute an overall factor to the state)
  to arrive at \(\mathbb{A}^{\mu} \simeq \mathbbm{1}_d\otimes A^{\mu}\)
  and \(\mathbb{B}^g \simeq \rho(g)\otimes V(g)\),
  up to normalization,\footnote{
    We use the symbol \(\simeq\) to denote equivalence up to normalization of the generated MPS.
  } or more formally
  \begin{equation}
    \label{eq:53}
    \ket{\Psi(\mathbb{A}, \mathbb{B})} = 
    \sqrt{|G|} \ket{\Psi(\mathbb{A}_{\rho}, \mathbb{B}_{\rho})}.
  \end{equation}
  Crucially \(\ket{\Psi(\mathbb{A}_{\rho}, \mathbb{B}_{\rho})}\) is an injective MPS.
  Incidentally, MNC phases, e.g.\ the Haldane phase for
  \(G = \mathbb{Z}_2\times \mathbb{Z}_2\), the Klein four-group,
  saturate the upper bound on inaccessible entanglement
  of SPT phases in one dimension~\cite{degrootInaccessibleEntanglementSymmetry2020},
  hinting at a possible connection.
\end{example}
The general case of multiple inequivalent projective irreps is more interesting.
To emphasize the role of the irreps in labeling the injective blocks of the MPS
(and to reduce clutter)
let us define
\begin{equation}
  \label{eq:51}
  \ket{\lambda} \coloneq \ket{\Psi(\mathbb{A}_{\lambda}, \mathbb{B}_{\lambda})}
\end{equation}
for any \(\lambda\in\Rep_{\mathbb{C}}^{\overline{\alpha}}G\).
Notice that in this notation the gauged MPS can be written as \(\ket{\tw{\overline{\alpha}}{L}}\)
and decomposes into
\begin{equation}
  \ket{\Psi(\mathbb{A}, \mathbb{B})} \equiv \ket{\tw{\overline{\alpha}}{L}} = \sum_{j=1}^r d_j \ket{\rho_j}.
\end{equation}
Here we cannot simply discard duplicate blocks,
since this would change the relative weight of the states
described by the individual blocks.
The states generated by the inequivalent projective irreps are pairwise orthogonal
\begin{equation}
  \label{eq:54}
  \braket{\rho_j|\rho_k} \sim \delta_{jk}
\end{equation}
as \(n\rightarrow\infty\)
and if the MPS generated by \(A\) is an RG--fixed point (zero correlation length), 
the orthogonality already holds at the level of the individual tensors.
The action of the emergent symmetry MPO on the individual blocks
also becomes more transparent in this notation:
\begin{equation}\label{eq:65}
  \mathcal{Q}_{\sigma} \ket{\lambda} = \frac{1}{d_{\sigma}}\ket{\sigma\otimes \lambda},
\end{equation}
indicating that the dual symmetry acts by permuting the injective blocks of the gauged MPS.
The statement that the gauged MPS is invariant under the action of the dual symmetry reads
\begin{align*}
  \mathcal{Q}_{\sigma} \ket{\tw{\overline{\alpha}}{L}}
  = \frac{1}{d_{\sigma}}\ket{\sigma\otimes \tw{\overline{\alpha}}{L}}
  = \frac{1}{d_{\sigma}}\ket{\mathbbm{1}_{d_{\sigma}}\otimes \tw{\overline{\alpha}}{L}}
  = \ket{\tw{\overline{\alpha}}{L}},
\end{align*}
where we used
\(\sigma\otimes \tw{\overline{\alpha}}{L} \cong \mathbbm{1}_{d_{\sigma}}\otimes \tw{\overline{\alpha}}{L}\),
cf.\ Lemma~\ref{lem:a16}.
It follows from the same lemma that the gauged MPS is the only state of the form \(\ket{\lambda}\)
invariant under \(\mathcal{Q}_{\sigma}\) for all \(\sigma \in \Rep G\).

\section{Mapping of Phases for Finite Abelian Groups}
\label{sec:mapping-phases}

In this section we will restrict the finite symmetry group \(G\) to be abelian
and analyze in detail how an MPS in the quantum phase \((G_{\mathrm{u}}, [\alpha])\)
behaves under twisted gauging,
where \(G_{\mathrm{u}}\leq G\) is the unbroken symmetry group
and the SPT part of the phase is characterized by
\([\alpha]\in \operatorname{H}^2(G_{\mathrm{u}}, \operatorname{U}(1))\).
We show (under rather general assumptions) that the phase of the gauged MPS
is independent of the particular MPS used to represent the initial phase
and we conclude that gauging lifts to a mapping on the phases themselves.

\subsection{Dual On-Site Symmetry}
\label{sec:dual-site-symmetry}

When \(G\) is abelian the categorical \(\Rep G\) symmetry,
discussed in Section~\ref{sec:repg-sym},
naturally reduces to a global on-site
\(\widehat{G} \coloneq \operatorname{Hom}(G, \operatorname{U}(1))\)
group symmetry.
Namely, viewing \(\chi\in \widehat{G}\) as an element
(and multiplication given by the tensor product) of \(\Rep G\),
we get
\begin{align}\label{eq:72}
  \mathcal{Q}_{\chi}\coloneq& \cdots\;
  \,\,\raisebox{-9.08113pt}{\includegraphics{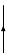}}\hspace{1.7em}
  \raisebox{-9.08113pt}{\includegraphics{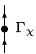}}
  \,\,\raisebox{-9.08113pt}{\includegraphics{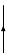}}\hspace{1.7em}
  \raisebox{-9.08113pt}{\includegraphics{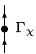}}
  \,\,\raisebox{-9.08113pt}{\includegraphics{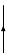}}\hspace{1.7em}
  \raisebox{-9.08113pt}{\includegraphics{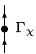}} \,
  \cdots,
\end{align}
equivalently \(\mathcal{Q} = (\operatorname{id}\otimes \Gamma)^{\otimes n}\), with
\begin{equation}
  \label{eq:50}
  \raisebox{-12.86032pt}{\includegraphics{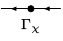}} \coloneq \sum_{g\in G} \chi(g) \,\raisebox{-2.5pt}{\includegraphics{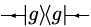}}.
\end{equation}
Hence, the state
\(\tw{\tau}{\mathcal{G}}\ket{\Psi(A)} = \ket{\Psi(\mathbb{A},\twb{\tau}{\mathbb{B}})}\)
is a translation invariant MPS,
generated by the tensor \({\mathbb{A}\!\twb{\tau}{\mathbb{B}}}\),
with local physical Hilbert space \(\mathcal{H}\otimes \mathbb{C}G\)
and on-site \(\widehat{G}\)-symmetry given by the local representation \(\mathrm{id} \otimes \Gamma\).

To characterize the phase after gauging we have to decompose
\(\ket{\Psi(\mathbb{A},\twb{\tau}{\mathbb{B}})}\)
into normal blocks and analyze how the operator \(\mathrm{id} \otimes \Gamma\),
generating the global symmetry \(\mathcal{Q}\),
acts on the tensors generating the blocks.
In detail we have to establish whether acting with \(\mathrm{id}\otimes \Gamma_{\chi}\) permutes the tensors
and/or how this action is represented on the virtual degrees of freedom of the individual tensors,
for any given \(\chi\in \widehat{G}\).

Already in Section~\ref{sec:repg-sym}
we mentioned the subtle issue of choice of dual symmetry,
i.e.\ given an isomorphism \(f\colon G \xrightarrow{\sim} H\),
it is equally valid to view \(\mathcal{Q}\) as a global \(\widehat{H}\)-symmetry.
Therefore, as part of a precise definition of gauging we also have to fix a choice of dual symmetry,
which we can do by specifying the isomorphism \(f\),
relating it to the reference symmetry \(\mathcal{Q}\) defined in Equation~\eqref{eq:72}.
In general, we should therefore analyze the action of \(\operatorname{id}\otimes \Gamma^f\) instead,
where
\begin{equation}
  \label{eq:71}
  \raisebox{-14.6211pt}{\includegraphics{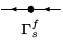}}
  \coloneq \sum_{g\in G} \underbrace{\hat{f}(s)(g)}_{s(f(g))} \,\raisebox{-2.5pt}{\includegraphics{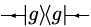}}
\end{equation}
for \(s \in \widehat{H}\) and a fixed but arbitrary isomorphism \(f\).
However, the general result can easily be recovered from the analysis of the specific choice of dual symmetry \(\mathcal{Q}\) made above, as we will discuss below.

For finite groups there always exists a (non-canonical) isomorphism \(G\cong \widehat{G}\).
Hence, we can also view the dual symmetry as a \(G\)-symmetry,
enabling us to compare directly the initial phase to the one resulting from the gauging procedure.
An interesting question to consider in this context is which (if any) phases are invariant
under gauging for certain choices of isomorphism \(f\).

\subsection{Main Theorem}
\label{sec:main-theorem}

In the Main Theorem~\ref{thm:main} we will describe how the phase
\((G_{\mathrm{u}}\leq G,[\alpha]\in\operatorname{H}^2(G_{\mathrm{u}}, \operatorname{U}(1)))\)
behaves under (twisted) gauging, assuming
that \(G_{\mathrm{u}}\) has a \emph{complement} in \(G\),
that is we assume there exists a subgroup \(G_{\mathrm{b}}\leq G\) such that
\begin{align}\label{eq:34}
  G \cong G_{\mathrm{b}} \times G_{\mathrm{u}}.
\end{align}
Note that if \(G_{\mathrm{b}}\) exists, it has to be isomorphic to \(G/G_{\mathrm{u}}\).
We call \(G_{\mathrm{b}}\) the \emph{broken} symmetry group.
Furthermore, in order to be able to analyze the twisted case,
we have to assume that the twist
\([\tau]\in \operatorname{H}^2(G,\operatorname{U}(1))\)
also factorizes according to Equation~\eqref{eq:34}
into \(\tau = \tau_{\mathrm{b}}\cdot \tau_{\mathrm{u}}\).
Stated more precisely, there exist
\(\tau_{\mathrm{b}} \in \operatorname{Z}^2(G_{\mathrm{b}}, \operatorname{U}(1))\)
and \(\tau_{\mathrm{u}} \in \operatorname{Z}^2(G_{\mathrm{u}}, \operatorname{U}(1))\)
such that
\begin{equation}\label{eq:98}
  \tau((b_1, u_1), (b_2, u_2)) = \tau_{\mathrm{b}}(b_1, b_2)\cdot \tau_{\mathrm{u}}(u_1, u_2)
\end{equation}
for all \(b_1, b_2 \in G_{\mathrm{b}}\) and \(u_1, u_2 \in G_{\mathrm{u}}\).

Before proceeding to the main theorem, we introduce two important lemmas,
the first being crucial for extracting the SPT part of the phase after gauging
and the second establishing the structure of a general matrix product state
with well defined phase \((G_{\mathrm{u}}, [\alpha])\),
making it easier to decompose the gauged MPS into normal blocks.
\begin{lemma}\label{lem:18}
  The slant product
  \(\imath\gamma\colon G_{\mathrm{u}}\rightarrow \widehat{G_{\mathrm{u}}}\),
  cf.\ Definition~\ref{def:slant-product-i},
  of \([\gamma] \in \operatorname{H}^2(G_{\mathrm{u}}, \operatorname{U}(1))\)
  induces a canonical isomorphism 
  \begin{equation}
    \label{eq:63}
    \widehat{\imath\gamma}\colon G_{\mathrm{u}}/K_{\gamma}
    \xlongrightarrow{\sim}
    \operatorname{ker} \Res ^{G_{\mathrm{u}}}_{K_{\gamma}},
  \end{equation}
  where \(K_{\gamma} \leq G_{\mathrm{u}}\) denotes the kernel of \(\imath\gamma\)
  and \(\Res \) denotes the canonical restriction\footnote{%
    \(\Res ^G_H \rho\equiv\rho|_H\) for \(\rho\in\Rep^{\alpha}_{\mathbb{C}}G\) and \(H\leq G\).
  } map.
\end{lemma}
\begin{proof}
  By the first group isomorphism theorem we get the induced isomorphism
  \begin{equation}
    \widehat{\imath\gamma}\colon G_{\mathrm{u}}/K_{\gamma} \equiv
    G_{\mathrm{u}}/ \operatorname{ker}\imath\gamma \xlongrightarrow{\sim} \operatorname{im}\imath\gamma
    \leq \operatorname{ker}\Res ^{G_{\mathrm{u}}}_{K_{\gamma}}
  \end{equation}
  and we show 
  \(\operatorname{im}\imath\gamma = \operatorname{ker}\Res ^{G_{\mathrm{u}}}_{K_{\gamma}}\)
  in Corollary~\ref{cor:can-iso}.
\end{proof}
\begin{lemma}\label{lem:main2}
  Let \(\ket{\Psi(A)}\) be a translation invariant MPS in the phase \((G_{\mathrm{u}}, [\alpha])\),
  with respect to a global on-site symmetry of the abelian group
  \(G = G_{\mathrm{b}}\times G_{\mathrm{u}}\).
  There exists a basis of the virtual Hilbert space \(\mathcal{V}\) such that:
  \begin{enumerate}
  \item The representation \((\mathcal{V}, V)\) of \(G_{\mathrm{b}}\times G_{\mathrm{u}}\)
    factorizes into
    \(\mathcal{V} \cong \mathbb{C}G_{\mathrm{b}}\otimes \mathcal{W}\)
    where \((\mathbb{C}G_{\mathrm{b}}, L) \in \Rep G_{\mathrm{b}}\)
    and \((\mathcal{W}, W)\in \operatorname{Rep}_{\mathbb{C}}^{\alpha} G_{\mathrm{u}}\).
    Graphically we depict the intertwiner property as
    \begin{equation}\label{eq:30}
      \raisebox{-14.9973pt}{\includegraphics{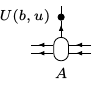}} = 
      \raisebox{-16.29784pt}{\includegraphics{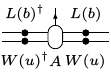}}
    \end{equation}
    for all \((b,u) \in G_{\mathrm{b}} \times G_{\mathrm{u}}\).
  \item The tensor \(A\) is \(G_{\mathrm{b}}\)-injective
    with respect to the representation \(\lambda\otimes\mathbbm{1} \in \Rep G_{\mathrm{b}}\), where
    \begin{equation}
      \lambda(b)
      \coloneq \sum_{c\in G_{\mathrm{b}}} \omega(b,c) \dyad{c}{c},
    \end{equation}
    with \(\omega\) an arbitrary but fixed non-degenerate bicharacter of \(G_{\mathrm{b}}\).
    Explicitly, this means that \(A\) has a purely virtual symmetry
    \begin{align}
      \label{eq:108}
      \raisebox{-14.9973pt}{\includegraphics{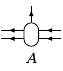}} = 
      \raisebox{-14.9973pt}{\includegraphics{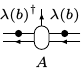}}
    \end{align}
    for all \(b\in G_{\mathrm{b}}\)
    and the transfer operator of \(A\) satisfies for \(n\in \mathbb{N}\)
    \begin{equation*}
      \label{eq:109}
      \left(\!\raisebox{-29.22366pt}{\includegraphics{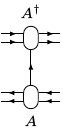}}\!\right)^{\hspace{-.3em}n} \hspace{-.45em}
      = \frac{1}{|G_{\mathrm{b}}|d_{\mathcal{W}}} \sum_{b\in G_{\mathrm{b}}}
      \raisebox{-13.91805pt}{\includegraphics{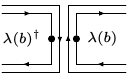}}
      + \mathcal{O}(\epsilon^n),
    \end{equation*}
    where \(0\leq \epsilon < 1\) is an upper bound on the magnitude of the subleading eigenvalues
    of the transfer operator \(\mathbb{T}(A, A)\).
  \end{enumerate}
  \emph{The proof is given in Appendix~\ref{sec:ansatz}.}
\end{lemma}
Notice that the set \(\left\{\lambda(b) \mid b \in G_{\mathrm{b}}\right\}\)
generates all diagonal matrices of size \(|G_{\mathrm{b}}|\).
This means that the virtual \(G_{\mathrm{b}}\)-symmetry of \(A\) is just a rephrasing of the block structure
introduced in Equation~\eqref{eq:91}.
Hence, we can view the upper horizontal legs of \(A\) as corresponding to the block degrees of freedom,
while the lower legs correspond to the degrees of freedom within the blocks.
\begin{theorem}[Main Theorem]\label{thm:main}
  An MPS in the phase \((G_{\mathrm{u}}, [\alpha])\),
  with respect to a global on-site symmetry of the abelian group
  \(G = G_{\mathrm{b}}\times G_{\mathrm{u}}\),
  will be mapped under (twisted) \(\tw{\tau}{\text{gauging}}\),
  assuming that \([\tau]\in \operatorname{H}^2(G,\operatorname{U}(1))\)
  also factorizes \(\tau = \tau_{\mathrm{b}}\cdot \tau_{\mathrm{u}}\) into broken and unbroken parts,
  cf.~Equation~\eqref{eq:98},
  to an MPS in the phase
  \begin{equation*}
    \left(
      \widehat{G_{\mathrm{b}}}\times \operatorname{ker}\Res ^{ G_{\mathrm{u}}}_{K_{\tau_{\mathrm{u}}/\alpha}}\leq \widehat{G}, \left[1\cdot\varphi^{\star}(\tau_{\mathrm{u}}/\alpha)\right]
    \right),
  \end{equation*}
  with respect to the dual on-site symmetry \(\mathcal{Q}\) of Equation~\eqref{eq:72},
  where \(\varphi^{\star}\) denotes the pullback along the inverse of the isomorphism introduced in Lemma~\ref{lem:18}
  and we view \([\tau_{\mathrm{u}}/\alpha]\) as an element of 
  \(\operatorname{H}^2(G_{\mathrm{u}}/K_{\gamma}, \operatorname{U}(1))\),
  cf.~Equation~\eqref{eq:42}.
  The phase after gauging does not depend on the broken part of the twist \([\tau_{\mathrm{b}}]\) 
  and the broken symmetry group after gauging
  is given by the kernel of the slant product of \([\tau_{\mathrm{u}}/\alpha]\).
\end{theorem}
Notice in particular that gauging of a phase without symmetry breaking
will produce a (potentially symmetry broken) phase, of which the SPT part is MNC.
Before proceeding to prove the Main Theorem we give some examples
serving to illustrate its applications in several simple scenarios.
\begin{example}\label{ex:trivial}
  Gauging of \((G_{\mathrm{u}}\leq G_{\mathrm{b}}\times G_{\mathrm{u}}, [\alpha] = 1)\)
  will produce \((G_{\mathrm{b}}\leq G_{\mathrm{b}}\times G_{\mathrm{u}}, [\alpha] = 1)\),
  since \(K_{\alpha} = G_{\mathrm{u}}\).
  In particular, gauging the trivial SPT results in the maximally symmetry broken phase and vice versa.
\end{example}
\begin{example}
  Let \(G_{\mathrm{u}}=G=\mathbb{Z}_2\times \mathbb{Z}_2\).
  The phase corresponding to 
  \(1 \neq [\alpha]\in \operatorname{H}^2(\mathbb{Z}_2\times \mathbb{Z}_2,\operatorname{U}(1))\),
  dubbed the Haldane phase, is MNC, hence
  \(\widehat{G}_{\mathrm{u}} = \operatorname{ker}\Res ^G_{\left\{ e \right\}} = \widehat{G}\).
  Since the Haldane phase is the only MNC phase of \(\mathbb{Z}_2\times \mathbb{Z}_2\),
  it has to be invariant under gauging.
\end{example}
The previous examples are well known.
What follows are two elementary but novel examples,
demonstrating that there is a non-trivial interplay between symmetry breaking and non-MNC SPT phases under gauging.
Notice that both \(\beta\) and \(\alpha^G\) will be non-MNC.
\begin{example}\label{ex:first-non-mnc}
  Let \(G_{\mathrm{u}}=G=\mathbb{Z}_2\times \mathbb{Z}_4\)
  and consider the (unique) non-trivial
  \([\beta]\in \operatorname{H}^2(\mathbb{Z}_2\times \mathbb{Z}_4,\operatorname{U}(1))\).
  Gauging of \((\mathbb{Z}_2\times \mathbb{Z}_4, [\beta])\)
  produces a phase with unbroken group
  \(\widehat{G}_{\mathrm{u}} = \operatorname{ker}\Res^G_{K_{\alpha}} \cong G/K_{\alpha} = \mathbb{Z}_2\times \mathbb{Z}_2\)
  and, as above, the SPT part given by 
  \(1 \neq [\alpha]\in \operatorname{H}^2(\mathbb{Z}_2\times \mathbb{Z}_2,\operatorname{U}(1))\).
\end{example}
\begin{example}
  Let \(G = \mathbb{Z}_2\times\mathbb{Z}_2\times\mathbb{Z}_2\) with one of the factors broken
  and \(1 \neq [\alpha]\in \operatorname{H}^2(\mathbb{Z}_2\times \mathbb{Z}_2,\operatorname{U}(1))\).
  Gauging of \((\mathbb{Z}_2\times\mathbb{Z}_2, [\alpha])\) results in the unbroken SPT phase
  described by \(\alpha^G\in \operatorname{Z}^2(G,\operatorname{U}(1))\),
  which extends \(\alpha\) trivially to the whole group.
\end{example}
The final example will be an important ingredient in generalizing the Kennedy-Tasaki transformation
to non-MNC SPT phases:
\begin{example}\label{ex:fsb}
  Let \(G_{\mathrm{u}}=G\) and \([\alpha] \in \operatorname{H}^2(G, \operatorname{U}(1))\).
  Twisted \(\tw{\alpha}{\text{gauging}}\) of \((G,[\alpha])\) yields the fully symmetry broken phase,
  since \(\operatorname{ker}\Res^G_G\) is the trivial group.
\end{example}
\begin{proof}[Proof of Main Theorem \ref{thm:main}]
  Let \(\ket{\Psi(A)}\) be a translation invariant MPS in the phase \((G_{\mathrm{u}}, [\alpha])\).
  We adopt the notation from Lemma~\ref{lem:main2} and
  begin by decomposing the gauged MPS into its injective blocks,
  similar to what we did in Section~\ref{sec:ghz-structure}.
  Since \(V\) is a projective representation of \(G\)
  with 2-cocycle \(\alpha^G\in \operatorname{H}^2(G,U(1))\), given by
  \(\alpha^G((b_1,u_1), (b_2, u_2)) \coloneq \alpha(u_1, u_2)\)
  for all \(b_1, b_2 \in G_{\mathrm{b}}\) and \(u_1, u_2 \in G_{\mathrm{u}}\),
  Lemma~\ref{lem:gauged-mps} says that the gauged state
  \(\tw{\tau}{\mathcal{G}} \ket{\Psi(A)}\)
  is an MPS generated by the tensors
  \begin{align}
    \mathbb{A}^{\mu} \coloneq \Stk{\raisebox{-11.23651pt}{\includegraphics{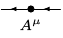}}}{\raisebox{-0pt}{\includegraphics{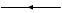}}},&&
    \twb{\tau}{\mathbb{B}}^{(b,u)} \equiv
    \twb{\tau}{\mathbb{B}}^g \coloneq \frac{1}{|G|}
    \Stk{\raisebox{-13.06595pt}{\includegraphics{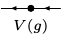}}}
        {\raisebox{-13.54901pt}{\includegraphics{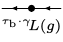}}},
  \end{align}
  where, to declutter the notation in the computations below, we introduced the abbreviation
  \(\gamma \coloneq \tau_{\mathrm{u}}\overline{\alpha} \in \operatorname{Z}^2(G_{\mathrm{u}}, \operatorname{U}(1))\).
  According to Lemma~\ref{lem:a15} the left regular representations satisfy
  \(\tw{\tau_{\mathrm{b}}\cdot \gamma}{L} \cong
  \tw{\tau_{\mathrm{b}}}{L}\otimes \tw{\gamma}{L}\),
  hence we can block decompose \(\twb{\tau}{\mathbb{B}}\) into
  \begin{equation}\label{eq:59}
    |G| \twb{\tau}{\mathbb{B}}^{(b,u)} \cong
    \mStk[.7cm]{{\raisebox{-13.06595pt}{\includegraphics{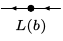}}}
                {\raisebox{-13.06595pt}{\includegraphics{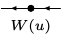}}}
                {\raisebox{-13.31592pt}{\includegraphics{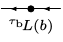}}}
                {\raisebox{-13.31592pt}{\includegraphics{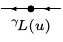}}}}
    \cong
    \mStk[.7cm]{{\raisebox{-13.31592pt}{\includegraphics{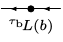}}}
                {\raisebox{-13.06595pt}{\includegraphics{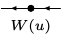}}}
                {\raisebox{-0pt}{\includegraphics{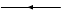}}}
                {\raisebox{-13.31592pt}{\includegraphics{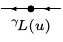}}}}
    \simeq \bigoplus_{\sigma}
    \sunderbrace{
    \mStk[.7cm]{{\raisebox{-13.31592pt}{\includegraphics{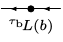}}}
                {\raisebox{-13.06595pt}{\includegraphics{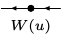}}}
                {\raisebox{-13.06595pt}{\includegraphics{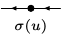}}}}
    }{\eqcolon |G|\mathbb{B}_{\sigma}},
  \end{equation}
  where we sum over all linearly inequivalent irreducible representations
  \(\sigma\) of \(G_{\mathrm{u}}\) with 2-cocycle \(\gamma\).
  Note that all \(\gamma\)-irreps share the same dimension \(d_{\gamma}\)
  and appear with equal multiplicity \(m_{\gamma} = d_{\gamma}\) in \(\tw{\gamma}{L}\),
  cf.~Lemma~\ref{lem:11}.
  Since the isomorphism
  \(L\otimes \tw{\tau_{\mathrm{b}}}{L} \cong \tw{\tau_{\mathrm{b}}}{L} \otimes \mathbbm{1}\),
  which, as we show in Appendix~\ref{sec:basis-transf},
  is explicitly given by conjugation with \(\sum_b \dyad{b}{b} \otimes \tw{\tau_{\mathrm{b}}}{L}(b)\),
  commutes with \(\mathbb{A}\),
  the corresponding (i.e.\ under the same basis transformation and up to the same multiplicities)
  block structure of \(\mathbb{A}\) reads
  \begin{equation}
    \label{eq:58}
    \mathbb{A} \simeq \bigoplus_{\sigma}
    \sunderbrace{
    \Stk{\raisebox{-14.9973pt}{\includegraphics{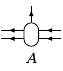}}}{\raisebox{-13.61035pt}{\includegraphics{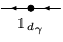}}}
    }{\eqcolon\mathbb{A}_{\sigma}}
  \end{equation}
  and we shall label the blocks of \(\mathbb{A}\) and \(\twb{\tau}{\mathbb{B}}\) by
  \(\mathbb{A}_{\sigma}\) and \(\mathbb{B}_{\sigma}\), respectively.
  To summarize, what we have shown thus far
  is that we can simultaneously block diagonalize \(\mathbb{A}\) and \(\twb{\tau}{\mathbb{B}}\),
  giving the decomposition
  \begin{equation}
    \label{eq:60}
    \ket{\Psi(\mathbb{A}, \tw{\tau}{\mathbb{B}})}
    =d_{\gamma}\sum_{\sigma} 
    \ket{\Psi(\mathbb{A}_{\sigma}, \mathbb{B}_{\sigma})}.
  \end{equation}
  We claim that this is a decomposition into injective and pairwise orthogonal
  globally symmetric matrix product states.
  The symmetry statement is clear and to show injectivity and orthogonality
  (cf.\ the discussion following Definition~\ref{def:transfer-operator})
  we compute the \(n^{\mathrm{th}}\) power of the mixed transfer operator,
  \begin{align*}
    \left(\!\raisebox{-38.18973pt}{\includegraphics{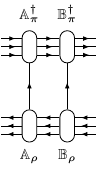}}\!\right)^{\hspace{-.3em}n} \hspace{-.45em}
    \stackrel{\scriptstyle\mathrm{(a)}}{=} \frac{1}{|G|}\sum_{\substack{b\in G_{\mathrm{b}}\\ u\in G_{\mathrm{u}}}}
    \mStk[1.5cm]{%
      {\raisebox{-0pt}{\includegraphics{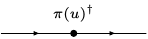}}}
      {\left(\!\raisebox{-32.06891pt}{\includegraphics{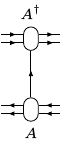}}\!\right)^{\hspace{-.3em}n}\hspace{-.3em}%
       \mStk[1.2cm]{%
         {\mStk[.14cm]{%
           {\raisebox{-0pt}{\includegraphics{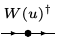}}}
           {\raisebox{-14.55594pt}{\includegraphics{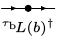}}}}}
         {\mStk[.14cm]{%
           {\raisebox{-0pt}{\includegraphics{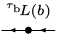}}}
           {\raisebox{-13.06595pt}{\includegraphics{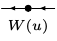}}}}}}}
      {\raisebox{-13.06595pt}{\includegraphics{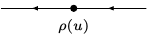}}}}\hspace{-.31em}&\\
    \stackrel{\scriptstyle\mathrm{(b)}}{=} \eqmakebox[op1]{\(\displaystyle\frac{1}{|G||G_{\mathrm{b}}|d_{W}}\!\!\sum_{\substack{b\in G_{\mathrm{b}}\\c\in G_{\mathrm{b}}\\ u\in G_{\mathrm{u}}}}\)\,}
    \mStk[.8cm]{%
      {\raisebox{-0pt}{\includegraphics{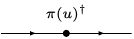}}}
      {\raisebox{-16.76329pt}{\includegraphics{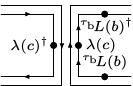}}}
      {\raisebox{-13.06595pt}{\includegraphics{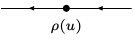}}}}
    &+ \mathcal{O}(\epsilon^n)\\
    \stackrel{\scriptstyle\mathrm{(c)}}{=} \eqmakebox[op1]{\(\displaystyle\frac{1}{|G|d_{W}}\sum_{u\in G_{\mathrm{u}}}\)}
    \mStk[.8cm]{%
      {\raisebox{-0pt}{\includegraphics{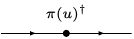}}}
      {\raisebox{-16.76329pt}{\includegraphics{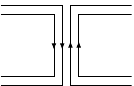}}}
      {\raisebox{-13.06595pt}{\includegraphics{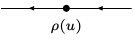}}}}
    &+ \mathcal{O}(\epsilon^n)\\
    \stackrel{\scriptstyle\mathrm{(d)}}{=} \eqmakebox[op1]{\(\displaystyle\frac{\delta_{\pi\rho}}{|G_{\mathrm{b}}|d_{\gamma}d_{W}}\)}
      \raisebox{-20.46228pt}{\includegraphics{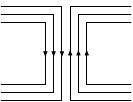}}
    &+ \mathcal{O}(\epsilon^n).
  \end{align*}
  The intertwiner property of \(A\) implies Equality (a),
  Lemmas~\ref{lem:main2} and~\ref{lem:bicharacter-com} yield (b) and (c), respectively,
  and Equality (d) is the Great Orthogonality Theorem for irreducible projective representations,
  cf.\ Theorem~\ref{thm:great-orthogonality}.

  Having established the block structure,
  we move to analyze the \(\widehat{G}\)-action on and within the blocks.
  The group \(G\) factoring into \(G_{\mathrm{b}}\times G_{\mathrm{u}}\) induces a canonical factorization
  \begin{equation}\label{eq:40}
    \widehat{G} \cong
    \widehat{G_{\mathrm{b}}} \times \widehat{G_{\mathrm{u}}},
  \end{equation}
  meaning we can analyze the action of
  \(\widehat{G_{\mathrm{b}}}\) and \(\widehat{G_{\mathrm{u}}}\) independently.
  Let us emphasize that \(\widehat{G_{\mathrm{u}}}\) denotes the dual group of \(G_{\mathrm{u}}\)
  and \emph{not} the unbroken part of the dual symmetry,
  which we would denote as \(\widehat{G}_{\mathrm{u}}\).
  
  Let \(\chi\in \widehat{G_{\mathrm{b}}} \leq \widehat{G}\)
  (with slight abuse of notation \(\chi(b,u) \equiv \chi(b)\))
  and choose \(c\in \widehat{G_{\mathrm{b}}}\) such that \(\chi = \omega(c, -)^{-1}\).
  Then
  \(\lambda(c)^{\dag} \tw{\tau_{\mathrm{b}}}{L}(b) \lambda(c) = \chi(b) \tw{\tau_{\mathrm{b}}}{L}(b)\)
  for all \(b\in G_{\mathrm{b}}\), cf.\ Lemma~\ref{lem:bicharacter-com}, and so
  \begin{equation}\label{eq:57}
    (\Gamma_{\chi}\twb{\tau}{\mathbb{B}}_{\sigma})^{(b,u)} = 
    \mStk[.7cm]{{\raisebox{-13.31592pt}{\includegraphics{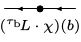}}}
                {\raisebox{-13.06595pt}{\includegraphics{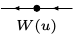}}}
                {\raisebox{-13.06595pt}{\includegraphics{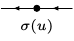}}}} = 
    \mStk[.7cm]{{\raisebox{-13.31592pt}{\includegraphics{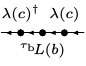}}}
                {\raisebox{-13.06595pt}{\includegraphics{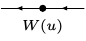}}}
                {\raisebox{-13.06595pt}{\includegraphics{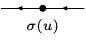}}}}
  \end{equation}
  for all \((b,u)\in G\).
  Since \(\lambda\) is a virtual symmetry of \(A\),
  we recognize that \(\mathcal{Q}_{\chi}\) does not
  permute between the injective blocks,
  hence \(\widehat{G_{\mathrm{b}}} \leq \widehat{G}_{\mathrm{u}}\),
  and since \(\chi \mapsto \lambda(c(\chi))\) is a linear representation of
  \(\widehat{G_{\mathrm{b}}}\), independent of the choice of \(\omega\),
  \(\widehat{G_{\mathrm{b}}}\) is contained in the kernel of the 2-cocycle
  describing the SPT part of the phase after gauging.
  
  Let \(\chi \in \widehat{G_{\mathrm{u}}} \leq \widehat{G}\).
  It is clear, immediately from the definition of \(\mathbb{B}_{\sigma}\),
  that the dual symmetry acts on the tensors representing the injective blocks as
  \begin{equation}
    \label{eq:64}
    \Gamma_{\chi}\mathbb{B}_{\sigma}
    = \mathbb{B}_{\chi\cdot\sigma},
  \end{equation}
  which mirrors what we derived in Equation~\eqref{eq:65}.
  This means that the \(\widehat{G_{\mathrm{u}}}\) action on blocks
  labeled by \(\gamma\)\nobreakdash-irreps is effectively given by
  the tensor product of representations modulo linear equivalence.
  Therefore, by Lemma~\ref{lem:12}, the unbroken part of \(\widehat{G_{\mathrm{u}}}\) 
  is given by the subgroup \(\operatorname{ker}\Res ^{G_{\mathrm{u}}}_{K_{\gamma}}\),
  where \(K_{\gamma}\leq G_{\mathrm{u}}\) denotes the kernel of the slant product \(\imath\gamma\),
  and the SPT part of the phase is determined by the 2-cocycle of 
  the representation \(\chi \mapsto S_{\chi}\) of \(\widehat{G_{\mathrm{u}}}\),
  appearing in the linear equivalence
  \begin{equation}
    \label{eq:67}
    (\chi\cdot\sigma)(u) = S_{\chi}^{\dag} \sigma(u) S_{\chi}, u\in G_{\mathrm{u}}.
  \end{equation}

  To disentangle notationally the virtual symmetry action within the injective blocks
  from the permutation action on the blocks,
  we can choose a slightly more explicit description
  of the irreducible representations appearing in Equation~\eqref{eq:59}.
  To this end, choose irreducible characters
  \(\chi_1, \ldots, \chi_{r_{\gamma}}\in \widehat{G_{\mathrm{u}}}\), such that
  \begin{equation}
    \label{eq:61}
    \widehat{G_{\mathrm{u}}}/
    \operatorname{ker}\Res ^{G_{\mathrm{u}}}_{K_{\gamma}}
    = \left\{ [\chi_1], \ldots, [\chi_{r_{\gamma}}] \right\}
  \end{equation}
  and fix an arbitrary irreducible \(\xi\in \Rep^{\gamma}_{\mathbb{C}}G_{\mathrm{u}}\).
  Lemma~\ref{lem:17} now readily implies that the block decomposition given in Equation~\eqref{eq:59}
  can equivalently be written as
  \begin{equation}\label{eq:62}
    \twb{\tau}{\mathbb{B}}^{(b,u)} \cong
    d_{\gamma}|G_{\mathrm{b}}| \bigoplus_{j = 1}^{r_{\gamma}} 
    \frac{1}{|G|}
    \mStk[.7cm]{{\raisebox{-13.31592pt}{\includegraphics{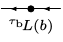}}}
                {\raisebox{-13.06595pt}{\includegraphics{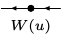}}}
                {\raisebox{-13.28815pt}{\includegraphics{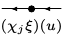}}}}.
  \end{equation}
  In other words, we consider characters of \(\widehat{G_{\mathrm{u}}}\) to
  be equivalent if they agree on \(K_{\gamma}\),
  the subgroup of \(G_{\mathrm{u}}\) where \(\gamma\) is trivial,
  and label the injective blocks in terms of representatives of these equivalence classes.
  
  To extract the SPT phase information, we shall determine how the action of
  \(\operatorname{ker} \Res ^{G_{\mathrm{u}}}_{K_{\gamma}}\)
  is represented on the virtual Hilbert space and assign the appropriate 2-cocycle.
  Let \(\chi\in \operatorname{ker} \Res ^{G_{\mathrm{u}}}_{K_{\gamma}}
  \leq \widehat{G_{\mathrm{u}}}\).
  Rewriting Equation~\eqref{eq:7},
  which is just the defining equation of the 2-cocycle \(\gamma\),
  in terms of \(\varphi \coloneq \widehat{\imath\gamma}^{-1}\),
  the inverse of the isomorphism introduced in Lemma~\ref{lem:18},
  yields
  \begin{equation}
    \label{eq:66}
    (\chi\cdot\xi)(u)
    = \xi(\varphi(\chi))^{-1}\,
      \xi(u)\,
      \xi(\varphi(\chi)),
  \end{equation}
  where, since \(\xi|_{K_{\gamma}}\) is proportional to the identity,
  we can view \(\xi\) as a representation of \(G_{\mathrm{u}}/K_{\gamma}\).
  Furthermore, since \(G\) is finite and abelian there is a canonical isomorphism
  \begin{equation}\label{eq:42}
    \operatorname{H}^2(G_{\mathrm{u}}/K_{\gamma}, \operatorname{U}(1)) \cong 
    \operatorname{H}^2(G_{\mathrm{u}}, \operatorname{U}(1))^{K_{\gamma}},
  \end{equation}
  induced by the quotient map, where the latter group denotes the subgroup
  of \(\operatorname{H}^2(G_{\mathrm{u}}, \operatorname{U}(1))\)
  whose elements have kernels containing \(K_{\gamma}\).
  This follows for instance from Theorem~28 in \cite[Chapter 6]{berkovicCharactersFiniteGroups1998},
  together with the fact that every linear character of \(K_{\gamma}\) can be extended to \(G_{\mathrm{u}}\).
  Hence, we can view \(\gamma\) to be a non-degenerate 2-cocycle on \(G_{\mathrm{u}}/K_{\gamma}\).
  
  Rewriting Equation~\eqref{eq:66}, under the above identifications,
  in terms of the MPS tensor \(\twb{\tau}{\mathbb{B}}\) yields the intertwiner relation
  \begin{equation}
    \raisebox{-17.3845pt}{\includegraphics{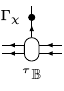}} = 
    \raisebox{-17.3845pt}{\includegraphics{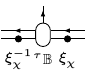}},
  \end{equation}
  where we abbreviated \(\xi(\varphi(\chi))\) as \(\xi_{\chi}\).
  Hence, we have shown that
  \(\operatorname{ker} \Res ^{G_{\mathrm{u}}}_{K_{\gamma}}\) is represented by
  \(\xi\circ\varphi \equiv \varphi^{\star}\xi\) on the virtual Hilbert space 
  and \(\varphi^{\star}\gamma\) is the corresponding 2-cocycle defining the SPT phase.
  
  Now we can combine the unbroken parts of both factors, giving us the unbroken symmetry group after gauging
  \begin{equation}
    \label{eq:39}
    \widehat{G}_{\mathrm{u}}
    = \widehat{G_{\mathrm{b}}}\times \operatorname{ker} \Res ^{G_{\mathrm{u}}}_{K_{\gamma}},
  \end{equation}
  with the SPT part of the phase given by the 2-cocycle \(\varphi^{\star}\gamma\), trivially extended from
  \(\operatorname{ker} \Res ^{G_{\mathrm{u}}}_{K_{\gamma}}\) to \(\widehat{G}_{\mathrm{u}}\).
  The broken part of the dual symmetry is of course defined to be 
  \(\widehat{G}_{\mathrm{b}} \coloneq \widehat{G}/\widehat{G}_{\mathrm{u}}\),
  which can be evaluated explicitly to yield
  \begin{align}
    \widehat{G}/\widehat{G}_{\mathrm{u}}
    \cong\frac{\widehat{G_{\mathrm{b}}}\times \widehat{G_{\mathrm{u}}}}
              {\widehat{G_{\mathrm{b}}}\times \operatorname{ker}\Res ^{G_{\mathrm{u}}}_{K_{\gamma}}}
    \cong\frac{\widehat{G_{\mathrm{u}}}}{\operatorname{ker} \Res ^{G_{\mathrm{u}}}_{K_{\gamma}}}
    \cong \widehat{K_{\gamma}},
  \end{align}
  where the first isomorphism is (basically) the isomorphism \eqref{eq:40},
  applied to both groups in the quotient,
  and the last isomorphism is just the statement of the first isomorphism theorem
  applied to the surjective restriction map
  \(\Res \colon \widehat{G_{\mathrm{u}}}\rightarrow \widehat{K_{\gamma}}\).
\end{proof}

\subsection{On the Choice of Dual Symmetry and Applications}
\label{sec:choice-dual-symmetry}

As we alluded to earlier, generalizing the results to different choices of dual symmetry,
generically of the form \(\Gamma^f_s\), cf.~Equation~\eqref{eq:71},
for a specific choice of isomorphism \(f \colon G\xrightarrow{\sim} H\), is trivial.
Broken and unbroken parts of the symmetry are given by
\begin{align}
  \label{eq:99}
  \widehat{H}_{\mathrm{b}} \coloneq \hat{f}^{-1}\widehat{G}_{\mathrm{b}},&&
  \widehat{H}_{\mathrm{u}} \coloneq \hat{f}^{-1}\widehat{G}_{\mathrm{u}},
\end{align}
where by \(\hat{f}^{-1}\) we denote the preimage under the dual map,
and the 2-cocyle is given by the pullback under the dual map,
appropriately restricted to the unbroken symmetry group:
\((\hat{f}|_{\widehat{H}_{\mathrm{u}}})^{\star}(1\cdot \varphi^{\star}(\tau/\alpha))\).
The relevance of the choice of dual symmetry is demonstrated in Example~\ref{ex:choice-of-iso}.
\begin{remark}
    Going from \(G_{\mathrm{u}}\leq G\) to \(G\cong G_{\mathrm{b}}\times G_{\mathrm{u}}\)
    involves a choice of complement, possibly influencing the phase before gauging,
    which has to be taken into account.
    Having fixed the isomorphism \(G\cong G_{\mathrm{b}}\times G_{\mathrm{u}}\), the isomorphism
    \(\widehat{G}\cong\widehat{G_{\mathrm{b}}}\times \widehat{G_{\mathrm{u}}}\)
    is canonical.
\end{remark}
The following example is meant to demonstrate how the specific choice of dual symmetry,
i.e.\ the choice of isomorphism \(f\) in Equation~\eqref{eq:71},
determines the permutation of MNC phases under gauging.
\begin{example}\label{ex:choice-of-iso}
  Let \(G_{\mathrm{u}}=G=\mathbb{Z}_3\times \mathbb{Z}_3\).
  Since both non-trivial 2-cocycles of \(\mathbb{Z}_3\times \mathbb{Z}_3\) are MNC,
  we know that gauging can only permute between the two, possibly trivially.
  We fix \([\alpha]\) by specifying its slant product
  \begin{align}
    \label{eq:76}
    (\imath_{(a_1, a_2)}\alpha)(b_1, b_2)=  \zeta^{a_1b_2-a_2b_1},
    &&\zeta \coloneq \mathrm{e}^{\frac{2\pi}{3}\mathrm{i}},
  \end{align}
  and write
  \(\operatorname{H}^2(\mathbb{Z}_3\times \mathbb{Z}_3,\operatorname{U}(1))
  = \left\{ 1, [\alpha], [\overline{\alpha}] \right\}\).
  To compare directly the phase before gauging with the phase after gauging,
  we want to view the gauged state as being \(G\)-symmetric.
  Hence, let us fix an isomorphism
  \(\chi\colon \mathbb{Z}_3\times \mathbb{Z}_3 \rightarrow \widehat{\mathbb{Z}_3} \times \widehat{\mathbb{Z}_3}\),
  defined by
  \begin{equation}
    \label{eq:78}
    \chi_{m_1, m_2}(b_1, b_2) \coloneq \zeta^{m_1b_2 + m_2 b_1},
  \end{equation}
  and choose the dual \(G\)-symmetry to be given by Equation~\eqref{eq:71}, for \(f\coloneq\chi^{-1}\).

  First we analyze the mapping of the phase \(\alpha\).
  Using \((\imath\alpha)^{-1}(\chi_{m_1, m_2}) = (m_1, -m_2)\),
  we can compute the slant product of the gauged 2-cocycle,
  which we recall is given by the pullback of
  \(\overline{\alpha}\) along \(\varphi \coloneq (\imath\alpha)^{-1}\),
  yielding
  \begin{equation}
    \label{eq:80}
    (\imath_{\chi_{m_1, m_2}}\varphi^{\star}\overline{\alpha})(\chi_{n_1, n_2}) = \zeta^{m_1 n_2 - m_2n_1}.
  \end{equation}
  Comparing the result with Equation~\eqref{eq:76},
  we recognize that under the isomorphism \eqref{eq:78} we can identify \([\varphi^{\star}\overline{\alpha}]\)
  with \([\alpha]\), i.e.\ \([\alpha]\) is invariant under gauging.

  Second, for the other phase \([\overline{\alpha}]\) one checks that
  \((\imath\overline{\alpha})^{-1}(\chi_{m_1, m_2}) = (-m_1, m_2)\),
  yielding for the computation of the pullback along \(\psi \coloneq {\imath\overline{\alpha}}^{-1}\) the result
  \begin{equation}
    \label{eq:81}
    (\imath_{\chi_{m_1, m_2}}\psi^{\star}\alpha)(\chi_{n_1, n_2}) = \zeta^{-m_1 n_2 + m_2n_1},
  \end{equation}
  which we interpret as \([\overline{\alpha}]\) getting mapped to \([\overline{\alpha}]\) under gauging.
  To summarize, for this choice of dual symmetry,
  both \([\alpha]\) and \([\overline{\alpha}]\) are invariant under gauging.

  On the other hand, if we choose the isomorphism 
  \(\vartheta\colon \mathbb{Z}_3\times \mathbb{Z}_3 \rightarrow \widehat{\mathbb{Z}_3} \times \widehat{\mathbb{Z}_3}\),
  where
  \begin{equation}
    \label{eq:83}
    \vartheta_{m_1, m_2}(b_1, b_2) \coloneq \zeta^{m_1b_2 - m_2b_1},
  \end{equation}
  instead of \(\chi\),
  we will find (directly from \(\vartheta = \imath\alpha\))
  that gauging maps \([\alpha]\) to \([\overline{\alpha}]\) and vice versa,
  i.e.\ gauging acts as the inverse on the two non-trivial SPT phases.
\end{example}
Having understood the effect a different choice of dual symmetry has on the phase after gauging,
we shall now discuss some more involved scenarios.
Example~\ref{ex:first-non-mnc} admits a straightforward generalization to larger prime numbers:
\begin{example}\label{ex:non-mnc-prime}
  Let \(G_{\mathrm{u}} = G = \mathbb{Z}_p\times \mathbb{Z}_{p^2}\) for \(p\) a prime number
  and \(1\neq[\alpha]\in \operatorname{H}^2(\mathbb{Z}_p\times \mathbb{Z}_{p^2},\operatorname{U}(1))\).
  The phase after gauging exhibits symmetry breaking
  with unbroken part \(\widehat{G}_{\mathrm{u}}\cong G/K_{\alpha} \cong \mathbb{Z}_p\times \mathbb{Z}_p\).
  Let \([\check{\alpha}]\in \operatorname{H}^2(G/K_{\alpha}, \operatorname{U}(1))\)
  denote the image of \(\alpha\) under the isomorphism \eqref{eq:42}
  and let the dual symmetry be given by 
  \(f = (\widehat{\imath\alpha})^{-1} = (\imath\check{\alpha})^{-1}\) in Equation~\eqref{eq:71}.
  In this sense, gauging maps \((G_{\mathrm{u}}, [\alpha])\)
  to \((\mathbb{Z}_p\times \mathbb{Z}_p, [\overline{\check{\alpha}}])\) and,
  since \(\operatorname{H}^2(\mathbb{Z}_p\times \mathbb{Z}_p,\operatorname{U}(1))\) is cyclic,
  any phase \((\mathbb{Z}_p\times \mathbb{Z}_{p^2}, [\beta]\neq 1)\) is mapped to 
  \((\mathbb{Z}_p\times \mathbb{Z}_p, [\overline{\check{\beta}}])\).
  Notice that any other choice of \(f\) that is still of the form \(\widehat{\imath\gamma}\)
  will yield the same result.
  However, as we saw already in the previous example,
  there are choices of \(f\) for which this does not hold.
\end{example}
\begin{example}
  When \(G_{\mathrm{u}} = G\), then for any non-degenerate \([\alpha]\) there exists a choice of dual symmetry,
  such that the phase \((G,[\alpha])\) gets mapped to \((G,[\overline{\alpha}])\)
  under gauging, namely \(f = (\imath\alpha)^{-1}\) in Equation~\eqref{eq:71}.
  In the case of degenerate \([\alpha]\) the same statement holds if \(G_{\mathrm{b}}\cong K_{\alpha}\).
\end{example}
\begin{example}[Generalized Kennedy-Tasaki transformation]
  For any 2-cocycle \(\alpha\),
  twisted \(\tw{\alpha}{\text{gauging}}\) of the trivial SPT phase yields,
  for an appropriately chosen dual symmetry,
  the phase \((G/K_{\alpha}, [\overline{\alpha}])\)
  and untwisted gauging of this phase in turn produces the SPT phase \((G,[\alpha])\).
  Combining this with Examples~\ref{ex:trivial} and~\ref{ex:fsb} we can construct,
  on the level of phases, the maps
  \begin{align*}
    \tw{\alpha}{\mathcal{G}}\colon \text{SPT}_{\alpha} &\longmapsto \text{FSB},\\
    \mathcal{G}\circ \tw{\alpha}{\mathcal{G}}\circ\mathcal{G}\colon \text{FSB}_{\phantom{\alpha}} &\longmapsto \text{SPT}_{\alpha},
  \end{align*}
  which have the expected behavior of a generalized Kennedy-Tasaki transformation and its inverse,
  also for non-MNC phases.
  For the spin-1 AKLT model the operator \(\mathcal{G}^{\alpha}\circ \mathcal{G}\)
  coincides with the Kennedy-Tasaki transformation \cite{lootensDualitiesOneDimensionalQuantum2023},
  which is consistent with our definition
  since the Haldane phase is invariant under \(\mathcal{G}\).
\end{example}

The following example goes beyond the assumptions of the Main Theorem.
It considers the case where the twist does not factor into broken and unbroken parts.
For this specific choice of group it can still be analyzed using the techniques we introduced above.
\begin{example}
  Let \(G_{\mathrm{u}} = \mathbb{Z}_n < \mathbb{Z}_n\times\mathbb{Z}_n = G\)
  and fix a twist \([\tau]\in \operatorname{H}^2(\mathbb{Z}_n\times \mathbb{Z}_n, \operatorname{U}(1))\)
  by specifying its slant product
  \begin{equation}
    \label{eq:11}
    \imath\tau(b_1, u_1)(b_2, u_2) = \zeta^{u_1b_2 - u_2b_1},
  \end{equation}
  where \(\zeta\) is an \(n^{\mathrm{th}}\) root of unity.
  Twisted \(\tw{\tau}{\text{gauging}}\) of an MPS \(\ket{\Psi(A)}\) in the phase
  \((\mathbb{Z}_n, 1)\)
  produces an MPS that, up to normalization, decomposes according to
  \begin{equation}
    \label{eq:21}
    \tw{\tau}{\mathcal{G}} \ket{\Psi(A)}
    = \sum_{k=1}^n \ket{\Psi(X^k\mathbb{A}X^{n-k}, \mathbb{B})}
  \end{equation}
  into injective and pairwise orthogonal MPSs, where
  \begin{align}
    \label{eq:22}
    X^k \mathbb{A} X^{n-k} \coloneq
    \raisebox{-15.03618pt}{\includegraphics{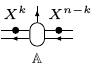}},&&
    \mathbb{B}^{(b,u)} \coloneq
    \Stk{\raisebox{-14.30594pt}{\includegraphics{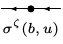}}}
        {\raisebox{-13.06595pt}{\includegraphics{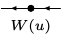}}}
  \end{align}
  and \(\sigma^{\zeta}(b,u) \coloneq X^b Z^u \) is the projective representation
  of Example~\ref{ex:sylvester},
  with \(X\) and \(Z\) the Sylvester shift and clock matrices,
  satisfying \(XZ = \zeta ZX\).
  To show this, just apply the basis transformation
  \begin{equation}
    \label{eq:41}
    S = \sum_{b\in \mathbb{Z}_n} X^b \otimes \operatorname{id}_{\mathbb{CZ}_n}\otimes \dyad{b}{b} 
  \end{equation}
  to the MPS tensors one gets after applying Lemma~\ref{lem:gauged-mps},
  with \(V(b,u) = X^b \otimes W(u)\).
  It is easy to check explicitly that the unbroken part of
  \(\widehat{G} = \widehat{\mathbb{Z}_n} \times \widehat{\mathbb{Z}_n}\)
  is given by elements of the form \((b,u)\mapsto \zeta^{kb}\) with \(k\in \mathbb{Z}_n\),
  while elements of the broken part are of the form \((b,u)\mapsto \zeta^{ku}\).
  Hence, the gauged MPS is still in the phase \((\mathbb{Z}_n, 1)\) and
  if the first factor of \(\mathbb{Z}_n\times \mathbb{Z}_n\) was initially broken,
  the second factor of the dual group will be broken after gauging.
\end{example}
\section{Discussion}
\label{sec:discussion}

In the Main Theorem~\ref{thm:main}, we comprehensively analyzed the behavior of
symmetry protected and symmetry broken phases under twisted gauging,
provided the symmetry group factors into unbroken and broken parts.

We observe that in the untwisted case the dual group of the kernel \(K_{\gamma} \equiv K_\alpha\),
which characterizes the degeneracy of the original 2-cocycle,
corresponds to the broken part of the symmetry group after gauging.
Conversely, since we have extended the induced 2-cocycle,
which is non-degenerate on \(\operatorname{ker} \Res ^{G_{\mathrm{u}}}_{K_{\gamma}}\)
by construction, trivially to
\(\widehat{G_{\mathrm{b}}}\times \operatorname{ker} \Res ^{G_{\mathrm{u}}}_{K_{\gamma}}\),
its kernel is just \(\widehat{G_{\mathrm{b}}}\),
i.e.\ the broken part of the original symmetry gives the degeneracy of the SPT after gauging.
This leads us to make the following conjecture.
\begin{conjecture}
  For abelian symmetry groups,
  degenerate SPT phases are dual to SSB phases under untwisted gauging.
\end{conjecture}
More precisely we mean that,
the degeneracy of the original SPT will become the broken symmetry after gauging,
while the degeneracy of the SPT after gauging will be given by the broken part of the original symmetry.

How restrictive is it to assume \(G\cong G_{\mathrm{b}}\times G_{\mathrm{u}}\)?
Clearly this can not hold in general, because already for \(\mathbb{Z}_2 \leq \mathbb{Z}_4\)
we have \(\mathbb{Z}_4 \ncong \mathbb{Z}_2 \times \mathbb{Z}_2\)
and similarly for \(p\) any prime number,
\(\mathbb{Z}_{p^k}\leq \mathbb{Z}_{p^l}\) gives a counterexample whenever \(k<l\).
Generally, the existence of a complement will depend on the specific subgroup of \(G\)
remaining unbroken, but there is a class of groups,
so called \emph{complemented} groups,
which admit a complement for every choice of subgroup.
A finite abelian group \(G\) is complemented if and only if
all of its Sylow subgroups are elementary, i.e.\ it is of the form
\begin{equation}
  G \cong \prod_{i=1}^n \mathbb{Z}_{p_i},
\end{equation}
where the \(p_i\) are (not necessarily distinct) prime numbers.
This is well known and follows directly from the fundamental theorem of finite abelian groups.
We see that complemented groups are precisely those groups,
which don't allow for the type of counterexample given above.
Hence, we have already proved the conjecture to hold for all complemented groups,
which we state as the following theorem:
\begin{theorem}
  For complemented abelian symmetry groups,
  degenerate SPT phases are dual to SSB phases under gauging.
\end{theorem}
However, since we expect the conjecture to hold more generally,
let us now briefly discuss the case where no complement to \(G_{\mathrm{u}}\) exists in \(G\).
The following example considers the very simplest such case.
\begin{example}
  Let \(G_{\mathrm{u}} = \mathbb{Z}_2 < \mathbb{Z}_4 = G\).
  Any representation of \(\mathbb{Z}_2\) is given by an involutory matrix~\(W\),
  hence, according to \cite{schuchClassifyingQuantumPhases2011}, the virtual representation
  \((\mathcal{V}, V)\in \Rep \mathbb{Z}_4\) of an MPS \(\ket{\Psi(A)}\) in the phase \((\mathbb{Z}_2, 1)\)
  has to be of the form
  \begin{equation}
    \label{eq:5}
    V(g) =
    \begin{pmatrix}
      0&W\\
      \mathbbm{1}&0
    \end{pmatrix}^g
  \end{equation}
  for \(g \in \mathbb{Z}_4\),
  and the MPS tensor \(A\) is \(\mathbb{Z}_2\)-injective
  (\(G_{\mathrm{b}}= \mathbb{Z}_2\)),
  with respect to the representation
  \begin{equation}
    \label{eq:10}
    \lambda(h) \otimes \mathbbm{1} = 
    \begin{pmatrix}
      \mathbbm{1}&0\\
      0&-\mathbbm{1}
    \end{pmatrix}^h
  \end{equation}
  for \(h \in \mathbb{Z}_2\),
  the proof of which can easily be adapted from the proof of Lemma~\ref{lem:main2}.
  Following the same arguments as in the proof of the Main Theorem,
  it is clear that the decomposition of \(L \in \Rep \mathbb{Z}_4\)
  into irreducible characters gives the decomposition of
  \(\mathcal{G}\ket{\Psi(A)}\) into injective blocks,
  of which only two are linearly independent.
  The action of the dual group permutes between the two blocks,
  hence, the gauged MPS remains in the phase \((\mathbb{Z}_2, 1)\).
\end{example}
This example can easily be generalized to the case 
\(\mathbb{Z}_{p^k} < \mathbb{Z}_{p^l}\),
which provides further support for our conjecture.

Finally, we make some general observations,
which could serve as a starting point for future investigations.
Note that a necessary condition for our MPS based approach to apply,
independent of the existence of a complement to \(G_{\mathrm{u}}\),
appears to be that the 2-cocycle
\(\alpha\in \operatorname{Z}^2(G_{\mathrm{u}}, \operatorname{U}(1))\),
describing the SPT part of the phase,
can by extended to a 2-cocycle of \(G\), i.e.\ there exists 
\(\alpha^G\in \operatorname{Z}^2(G, \operatorname{U}(1))\),
which we call \emph{induced} 2-cocycle,
such that the restriction of \([\alpha^G]\) to \(G_{\mathrm{u}}\) is equal to \([\alpha]\).
An MPS tensor describing a state in such a phase
cannot be an intertwiner between virtual and physical representations if no induced 2-cocycle exists,
which in turn prevents the derivation of Equations \eqref{eq:3}.
When a complement to \(G_{\mathrm{u}}\) exists, one can just extend the 2-cocycle trivially,
as we did in the previous section, however, in general not every \([\alpha]\in \operatorname{H}^2(G_{\mathrm{u}}, \operatorname{U}(1))\)
can be extended to the whole group \(G\), as the following example demonstrates.
\begin{example}
  Let \(G_{\mathrm{u}} =\mathbb{Z}_2\times \mathbb{Z}_2 < \mathbb{Z}_2\times \mathbb{Z}_4 = G\)
  and \(1\neq[\alpha]\in \operatorname{H}^2(\mathbb{Z}_2\times \mathbb{Z}_2,\operatorname{U}(1))\).
  Consider the projective representation
  \begin{align}
    \label{eq:69}
    \rho\colon\eqmakebox[eq:69:1]{\(\mathbb{Z}_2\times \mathbb{Z}_4 \)} &\longrightarrow \eqmakebox[eq:69:2]{\(\operatorname{GL}(2, \mathbb{R})\)}\\\nonumber
    \eqmakebox[eq:69:1]{\((a, b)\)} &\longmapsto\eqmakebox[eq:69:2]{\(\sigma_{\mathrm{x}}^a\sigma_{\mathrm{z}}^b\),}
  \end{align}
  where \(\sigma_{\mathrm{x}}\) and \(\sigma_{\mathrm{z}}\) are the usual Pauli matrices,
  and let \(\beta\in \operatorname{Z}^2(\mathbb{Z}_2\times \mathbb{Z}_4,\operatorname{U}(1))\)
  denote its 2-cocycle.
  It is straightforward to check, e.g., by computing its slant product,
  that \([\beta]\) is the unique non-trivial element of \(\operatorname{H}^2(\mathbb{Z}_2\times \mathbb{Z}_4, \operatorname{U}(1))\),
  but since \(\rho\) restricted to \(\mathbb{Z}_2\times \mathbb{Z}_2\) is linear,
  \([\beta]\) must restrict to \(1 \in \operatorname{H}^2(\mathbb{Z}_2\times \mathbb{Z}_2, \operatorname{U}(1))\),
\end{example}

If we assume that \(\alpha\in \operatorname{Z}^2(G_{\mathrm{u}}, \operatorname{U}(1))\)
can be extended to a 2-cocycle \(\alpha^G\) of the entire group,
then the projective representation on the virtual Hilbert space of any MPS in the phase
\((G_{\mathrm{u}}, [\alpha])\) will be equivalent to an induced representation
\(\Ind  W \in \Rep^{\alpha^G}_{\mathbb{C}}G\)
for some \(W \in \Rep ^{\alpha}_{\mathbb{C}} G_{\mathrm{u}}\).

If \(\alpha\) is non-degenerate, we have
\(\Res ^G_{G_{\mathrm{u}}}\rho\cong \Res ^G_{G_{\mathrm{u}}}\sigma\)
for any irreducible \(\overline{\alpha^G}\)-representations \(\sigma\) and \(\rho\) of G
and hence
\begin{align*}
  \Ind W\otimes \rho
  &\cong \Ind  \left(W\otimes \Res \rho\right)\\
  &\cong \Ind  \left(W\otimes \Res \sigma\right)\\
  &\cong \Ind W\otimes \sigma,
\end{align*}
where \(\Res \equiv \Res ^G_{G_{\mathrm{u}}}\)
and \(\Ind \equiv \Ind ^G_{G_{\mathrm{u}}}\).
It is straightforward to show that this isomorphism commutes with \(\mathbb{A}\),
so that here again the decomposition of the \(\overline{\alpha^G}\)-regular representation
into irreducibles will not contribute to symmetry breaking after the gauging procedure.
Also, the fact that \(\Res \chi\rho = \Res \rho\) for
any \(\chi \in \widehat{G}\) in the kernel of \(\Res ^G_{G_u}\)
suggest that \(G_{\mathrm{b}}\) gives the degeneracy after gauging.
The cases where the initial phase features symmetry breaking or
\(\alpha\) is degenerate are more involved an will be left for future investigations.

\section{Acknowledgements}
D.B.\ thanks Andr\'as Moln\'ar for many enlightening discussions.
This research has been funded in part by the  European Union’s Horizon 2020
research and innovation programme through Grant No.\ 863476 (ERC-CoG SEQUAM),
and the Austrian Science Fund (FWF) through Grants
\href{https://doi.org/10.55776/COE1}{10.55776/COE1},
\href{https://doi.org/10.55776/F71}{10.55776/F71}, and the FWF Quantum Austria
Funding Initiative project ``Entanglement Order Parameters''
(\href{https://doi.org/10.55776/P36305}{10.55776/P36305}), funded through the
European Union (NextGenerationEU).
J.G.R.\ also acknowledges funding by the FWF Erwin Schrödinger Program (\href{https://doi.org/10.55776/J4796}{10.55776/J4796}).
\end{multicols}

\newpage
\begin{appendices}
\smallskip
\begin{multicols}{2}
\section{Projective representations of finite groups}
\label{sec:proj-repr}
We use Appendix \ref{sec:proj-repr} to fix notation
and to remind the reader of several important results
from the theory of projective representations of finite groups.
Everything that is contained in this section
is already well established in the literature and
we make no claim to originality.
\subsection{Twisted group algebra}
While we will exclusively consider representations over the complex numbers \(\mathbb{C}\),
all statements remain valid for any algebraically closed field
of characteristic not dividing the order of the group.
Traditionally, a projective representation of a group \(G\)
on a complex vector space \(W\) is defined to be a homomorphism
\begin{equation}
  G\longrightarrow \operatorname{PGL}(W)= \operatorname{GL}(W)/ \mathbb{C}^{\times},
\end{equation}
or more concretely, a projective representation of \(G\) on \(W\)
can be thought of as a collection of linear maps
\(\{\rho(g)\in \operatorname{GL}(W)\mid g \in G\}\),
that satisfy
\begin{equation}
  \rho(g)\rho(h) = \alpha(g,h) \rho(gh),
\end{equation}
with \(\alpha\in \mathrm{Z}^2(G, \mathrm{U}(1))\) a so called 2-cocycle.
In modern language,
the (projective) representation theory of groups
is formulated in terms of representations of group algebras.
\begin{definition}[Twisted group algebra]
  Given a 2-cocycle \(\alpha\in \mathrm{Z}^2(G, \mathrm{U}(1))\)
  of the finite group \(G\),
  we define the twisted group algebra
  \(\mathbb{C}^{\alpha}G\)
  of \(G\) over \(\mathbb{C}\) with respect to \(\alpha\),
  as the free complex vector space on the basis
  \(\left\{ \ket{g} \right\}_{g\in G}\),
  with multiplication defined on the basis vectors as
  \begin{equation}
    \ket{g}\cdot\ket{h} = \alpha(g,h) \ket{gh}.
  \end{equation}
\end{definition}
\begin{remark}
  The 2-cocycle condition of \(\alpha\) guarantees the associativity of \(\mathbb{C}^{\alpha}G\).
  The algebra structure of \(\mathbb{C}^{\alpha}G\)
  only depends on the cohomology class \([\alpha] \in \operatorname{H}^2(G, \operatorname{U}(1))\)
  of \(\alpha\)
  and not the specific choice of representative.
  Indeed, if \(\alpha^{\prime}\)
  and \(\alpha\) only differ by a coboundary \(d\gamma\), then
  \(\mathbb{C}^{\alpha}G \cong \mathbb{C}^{\alpha^{\prime}}G\)
  as complex algebras,
  since \(\ket{g} \mapsto \gamma(g) \ket{g}\) produces an isomorphism.
\end{remark}
\begin{definition}[Projective representations]
  Let \(\alpha\in \mathrm{Z}^2(G, \mathrm{U}(1))\).
  We define a \emph{projective representation}
  of \(G\) with 2-cocycle \(\alpha\),
  also called  \(\alpha\)-projective representation,
  or \(\alpha\)-representation for short,
  to be a left \(\mathbb{C}^{\alpha}G\)-module.
  We denote the category of all \(\alpha\)-projective representations of \(G\) over \(\mathbb{C}\) as
  \(\Rep^{\alpha}_{\mathbb{C}}G\).
\end{definition}
\begin{remark}
  Clearly, \(\alpha\)-projective representations in the earlier sense
  are in one to one correspondence with linear representations of \(\mathbb{C}^{\alpha}G\),
  i.e.\ \(\mathbb{C}^{\alpha}G\)-modules.
  We will often abuse notation and denote representations
  simply by their corresponding linear maps,
  as is common practice in the physics literature.
\end{remark}
\begin{definition}[Linear and projective equivalence]
  Two \(\alpha\)-projective representations \(\rho\) and \(\rho^{\prime}\)
  are said to be \emph{linearly equivalent},
  denoted \(\rho \cong \rho^{\prime}\),
  if they are isomorphic as \(\mathbb{C}^{\alpha}G\)-modules.
  In terms of linear maps this means that there exists an isomorphism \(V\)
  of the underlying vector spaces,
  such that for all \(g\in G\)
  \begin{equation}
    \rho(g) = V^{-1}\rho^{\prime}(g)V.
  \end{equation}
  We will often abuse notation and write \(\rho(g) \cong \rho^{\prime}(g)\) instead.
  On the level of linear maps there is another notion of equivalence one can define.
  Two projective representations \(\rho\) and \(\sigma\),
  possibly with different 2-cocycles,
  are said to be \emph{projectively equivalent},
  denoted \(\rho \sim \sigma\),
  if there exists a map
  \(\phi\colon G \rightarrow \mathbb{C}^{\times}\),
  such that \(\phi\cdot\rho\) and \(\sigma\) are linearly equivalent.
\end{definition}
\begin{remark}
  Clearly, the 2-cocycles of two projectively equivalent representations
  can only differ by a coboundary, hence share the same equivalence class.
  Since their respective twisted group algebras are thus isomorphic one
  could also formulate projective equivalence in terms of isomorphic modules.
  Notice however that the notion of projective equivalence is not appropriate for
  describing decompositions into subrepresentations,
  as it is not possible to define addition of representations with differing 2-cocycles,
  even when they belong to the same equivalence class.
\end{remark}
\begin{definition}[Regular projective representations]\label{def:regular}
  Let \(\alpha\in \mathrm{Z}^2(G, \mathrm{U}(1))\).
  The left regular \(\mathbb{C}^{\alpha}G\)-module
  induces the \(\alpha\)-projective left regular representation \(L^{\alpha}\),
  which, as linear maps on the underlying vector space \(\mathbb{C}G\),
  is given by
  \begin{equation}\label{eq:15}
    \tw{\alpha}{L}(g) = \sum_{h\in G} \alpha(g, h) \dyad{gh}{h}.
  \end{equation}
  Multiplication by \(\ket{g}^{-1}\) from the right exhibits
  \(\mathbb{C}^{\alpha^{-1}}G\) as a left \(\mathbb{C}^{\alpha}G\)-module,
  which we call the \(\alpha\)-projective right regular representation \(R^{\alpha}\).
  As linear maps on \(\mathbb{C}G\) it is given by
  \begin{equation}\label{eq:17}
    \tw{\alpha}{R}(g) = \sum_{h\in G} \alpha(h, g) \dyad{h}{hg}.
  \end{equation}
\end{definition}
\begin{remark}
  If \(\alpha\) is the trivial 2-cocycle we recover the standard
  left and right regular representations.
  In general, \(\tw{\alpha}{L}\) will commute with \(\tw{\beta}{R}\)
  only if \([\alpha] = [\beta]^{-1}\).
  Note that the right regular \(\mathbb{C}^{\alpha}G\)-module
  is a left \((\mathbb{C}^{\alpha}G)^{\mathrm{op}}\)-module
  and in general not isomorphic to a 
  left \(\mathbb{C}^{\alpha}G\)-module.
\end{remark}
We finish this section by stating several well known and important results,
which are completely analogous to the corresponding results concerning linear representations
and can be proven analogously.
\begin{theorem}[Maschke]\label{thm:maschke}
  \(\mathbb{C}^{\alpha}G\) is semisimple.
\end{theorem}
\begin{remark}
  This implies that every projective representation of a finite group decomposes into
  a direct sum of irreducible representation.
\end{remark}
\begin{deflemma}\label{lem:4}
  The set of classes of linearly inequivalent irreducible representations of \(G\)
  with 2-cocycle \(\alpha\) is finite. 
  We let \(r_{\alpha}\) denote the cardinality of this set.
\end{deflemma}
\begin{lemma}\label{lem:5}
  \(r_{\alpha} = \operatorname{dim}_{\mathbb{C}} \operatorname{Z}(\mathbb{C}^{\alpha}G)\),
  where \(\operatorname{Z}(-)\) denotes taking the center of an algebra or ring.
\end{lemma}
Lemmas~\ref{lem:4} and~\ref{lem:5} are stated and proved in, e.g.,
Chapter 6 of the textbook \cite{berkovicCharactersFiniteGroups1998},
as Corollary 11 and Lemma 14 respectively.
\subsection{Character theory}\label{sec:character-theory}
The character theory for projective representations of finite groups
develops mostly analogous to the linear case.
Choosing \(\alpha\) to be the trivial 2-cocycle,
reduces all statements to the character theory for linear representations.
For details see, e.g., \cite{chengCharacterTheoryProjective2015}
or that author's accompanying notes of the same name.
Proofs, which we omit here, can be found there.
\begin{definition}[Unitary 2-cocycles]
  We say that the 2-cocycle \(\alpha\) is \emph{unitary}
  if there exists a number \(N\in \mathbb{N}\),
  such that \(\alpha^N=1\).
\end{definition}
\begin{definition}[Normalized 2-cocycle]
  We say that the 2-cocycle \(\alpha\) of \(G\) is \emph{normalized}
  if for all \(g\in G\)
  \begin{equation}
    \label{eq:29}
    \alpha(g, e) = \alpha(e, g) = 1,
  \end{equation}
  where \(e\) denotes the unit element of \(G\).
\end{definition}
\begin{remark}
  One can show that for any 2-cocycle \(\alpha\),
  there exists a unitary 2-cocycle
  \(\alpha^{\prime}\) with \([\alpha] = [\alpha^{\prime}]\)
  and if \(\alpha\) is normalized, \(\alpha^{\prime}\)
  can be chosen to be normalized as well,
  hence we will restrict our discussion to unitary 2-cocycles.
  This is not strictly necessary, but simplifies some arguments.
\end{remark}
\begin{definition}[Projective characters]
  Let \(\rho\) be an \(\alpha\)-projective representation of \(G\).
  We call \(\chi_{\rho}\coloneq \operatorname{tr}\circ\, \rho\)
  the \emph{character} of \(\rho\).
  We say a function \(\chi\colon G\rightarrow \mathbb{C}\) is an
  (irreducible) \(\alpha\)-character,
  if there exists an (irreducible) \(\alpha\)-representation \(\rho\),
  such that \(\chi = \chi_{\rho}\).
\end{definition}
\begin{lemma}
  Let \(\chi\) be a unitary \(\alpha\)-character of \(G\).
  For all \(g,h\in G\) we have
  \begin{align}
    \chi(hgh^{-1}) &= \frac{\alpha(h, h^{-1})}{\alpha(h, gh^{-1})\alpha(g, h^{-1})} \chi(g),\\
    \chi(g^{-1}) &= \alpha(g, g^{-1}) \overline{\chi}(g).
  \end{align}
\end{lemma}
\begin{lemma}[Schur orthogonality]\label{lem:orhtogonal-characters}
  Irreducible unitary \(\alpha\)-projective characters
  form an orthonormal systems with respect to the scalar product
  \begin{equation}
    \label{eq:20}
    \braket{\chi, \eta} \coloneq \frac{1}{|G|}\sum_{g\in G} \overline{\chi(g)} \eta(g).
  \end{equation}
\end{lemma}
\begin{corollary}\label{cor:character-multiplicity}
  Suppose \((\mathcal{H}, U)\in \Rep^{\alpha}_{\mathbb{C}} G\) decomposes into irreducible representations according to
  \begin{equation}
    \label{eq:23}
    \mathcal{H} \cong \mathcal{W}_1 \oplus \cdots \oplus \mathcal{W}_k.
  \end{equation}
  Then \(\braket{\operatorname{tr}\circ\,W, \operatorname{tr}\circ\, U}\) 
  equals the number of modules \(\mathcal{W}_i\) in the decomposition of \(\mathcal{H}\)
  isomorphic to \(\mathcal{W}\),
  for any irreducible \((\mathcal{W}, W)\in \Rep_{\mathbb{C}}^{\alpha}\).
\end{corollary}
\begin{corollary}\label{cor:isomorphic-characters}
  Two unitary \(\alpha\)-projective representations over \(\mathbb{C}\) are isomorphic,
  if and only if they have the same character.
\end{corollary}
\begin{deflemma}
  Let \(\tw{\alpha}{L}\) be the left regular \(\alpha\)-projective representation 
  (cf.\ Definition~\ref{def:regular})
  and \(\ell^{\alpha}\coloneq \operatorname{tr}\circ \tw{\alpha}{L}\)
  its character.
  We have 
  \begin{equation}
    \ell^{\alpha}(g) = |G| \,\delta_{g, e}.
  \end{equation}
\end{deflemma}
\begin{proof}
  In the standard basis of \(\mathbb{C}G\),
  \(\tw{\alpha}{L}\) is given by Equation~\eqref{eq:15}.
  Since we assume \(\alpha\) to be normalized,
  \(\ell^{\alpha}(e)=|G|\) follows immediately
  and if \(g\neq e\) then \(gh \neq h,\) for all \(h \in G,\) implies that
  all diagonal entries of \(\tw{\alpha}{L}(g)\) are zero,
  hence \(\ell^{\alpha}(g) = 0\).
\end{proof}
\begin{corollary}\label{cor:reg-char}
  From the proof of the previous lemma we also see that
  \begin{equation}
    \operatorname{tr}\left\{ D \tw{\alpha}{L}(g) \right\}
    = \delta_{g,e}\operatorname{tr} D,
  \end{equation}
  for any matrix \(D\) which is diagonal
  in the standard basis of \(\mathbb{C}G\).
\end{corollary}
\begin{corollary}\label{cor:tr-reg-rep}
  The left regular \(\alpha\)-representation
  \(\tw{\alpha}{L}\) of \(G\)
  is (up to isomorphism) uniquely defined by the property
  \begin{equation}
    \operatorname{tr}(\tw{\alpha}{L}(g)) = \ell^{\alpha}(g) = \left| G \right| \delta_{g, e}.
  \end{equation}
\end{corollary}
\begin{corollary}\label{cor:regular-decomp}
  Denote by \(\rho_1, \ldots, \rho_r\) the distinct,
  i.e.\ pairwise linearly inequivalent,
  irreducible projective representations of \(G\) with 2-cocycle \(\alpha\),
  then we have the decomposition
  \begin{equation}
    \label{eq:89}
    \tw{\alpha}{L} \cong
    \bigoplus_{i=1}^r \mathbbm{1}_{|\rho_i|}\otimes \rho_i.
  \end{equation}
\end{corollary}
\begin{proof}
  Since
  \(\braket{\ell^{\alpha}, \operatorname{tr}\circ \rho_i} = |\rho_i|\),
  the result follows from Corollary~\ref{cor:character-multiplicity}.
\end{proof}
\begin{lemma}\label{lem:a15}
  Assume \(\alpha \in \operatorname{Z}^2(G_1 \times G_2, \operatorname{U}(1))\)
  decomposes as \(\alpha = \alpha_1\alpha_2\), i.e.\ 
  \begin{equation}
    \alpha((g_1, g_2), (h_1, h_2)) = \alpha_1(g_1, h_1)\alpha_2(g_2, h_2),
  \end{equation}
  for all \(g_i,h_i \in G_i\),
  where
  \(\alpha_i\in \operatorname{H}^2(G_i, \operatorname{U}(1))\),
  then \(\tw{\alpha}{L}\) is isomorphic to \(\tw{\alpha_1}{L}\otimes \tw{\alpha_2}{L}\).
\end{lemma}
\begin{proof}
  By assumption,
  \(\tw{\alpha_1}{L}\otimes \tw{\alpha_2}{L}\) is a representation of \(G\) with 2-cocycle \(\alpha\)
  and according to Corollary~\ref{cor:tr-reg-rep}
  the regular \(\alpha\)-representation is (up to isomorphism) uniquely defined by its character,
  so we compute
  \pagebreak
  \begin{align*}
    \operatorname{tr}\big(\tw{\alpha_1}{L}(g_1) \otimes \tw{\alpha_2}{L}(g_2)\big)
    &= \operatorname{tr}\tw{\alpha_1}{L}(g_1) \operatorname{tr}\tw{\alpha_2}{L}(g_2)\\
    &= \ell^{\alpha_1}(g_1) \ell^{\alpha_2}(g_2)\\
    &= \left|G_1\right| \delta_{g_1,e} \left|G_2\right| \delta_{g_2,e}\\
    &= \left|G_1\times G_2\right| \delta_{(g_1,g_2), e} \\
    &= \ell^{\alpha}(g_1, g_2). \qedhere
  \end{align*}
\end{proof}
\begin{lemma}\label{lem:a16}
  The left regular \(\alpha\)-representation is invariant under
  taking tensor products with linear representations,
  i.e.\ \(-\otimes L^{\alpha} \sim L^{\alpha}\),
  and is essentially the only \(\alpha\)-representation with this property.
  Formulated more precisely, for any \(\lambda \in \Rep_{\mathbb{C}}^{\alpha}G\),
  satisfying
  \begin{equation}\label{eq:96}
    \forall \rho \in \Rep G\colon
    \rho \otimes \lambda \cong \mathbbm{1}_{|\rho|}\otimes \lambda,
  \end{equation}
  there exist \(m, n\in \mathbb{N}\) such that
  \begin{equation}
    \mathbbm{1}_m\otimes \lambda \cong \mathbbm{1}_n\otimes L^{\alpha}.
  \end{equation}
\end{lemma}
\begin{proof}
  Let \(h\in G\) such that \(\operatorname{tr}\lambda(h)\neq 0\).
  From Equation~\eqref{eq:96} we see that \(\operatorname{tr}\rho(h) = |\rho|\),
  i.e.\ \(\rho(h) = 1\) for all unitary \(\rho \in \Rep G\),
  implying \(h = e\).
  We conclude \(\operatorname{tr}\lambda(g) = |\lambda|\, \delta_{g,e}\)
  and, by Corollary~\ref{cor:tr-reg-rep}, the result follows.
\end{proof}
\begin{theorem}[Great Orthogonality Theorem]\label{thm:great-orthogonality}
  Let \((\mathcal{H}_{\rho}, \rho)\) and \((\mathcal{H}_{\sigma}, \sigma)\)
  be equal or inequivalent irreducible \(\alpha\)-projective
  representations of the finite group \(G\). Then
  \begin{equation}
    \label{eq:38}
    \sum_{g\in G}\,
    \mStk[.8cm]{{\raisebox{-14.30594pt}{\includegraphics{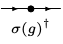}}}
                {\raisebox{-13.06595pt}{\includegraphics{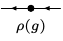}}}}
    = \frac{|G|\delta_{\rho \sigma}}{d_{\rho}}
    \raisebox{-9.08113pt}{\includegraphics{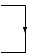}}\,\raisebox{-9.08113pt}{\includegraphics{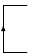}}\,.
  \end{equation}
\end{theorem}
\begin{proof}
  Define
  \begin{equation}
    \label{eq:36}
    M \coloneq \sum_{g\in G} \rho(g) X \sigma(g)^{\dag}
  \end{equation}
  for \(X\) an arbitrary linear map from \(\mathcal{H}_{\sigma}\) to \(\mathcal{H}_{\rho}\).
  The calculation
  \begin{align*}
    \rho(h)M
    &= \sum_{g\in G} \alpha(h,g) \rho(hg)X\sigma(g)^{\dag}\\
    &= \sum_{g\in G} \alpha(h,h^{-1}g) \rho(g)X\sigma(h^{-1}g)^{\dag}\\
    &= \sum_{g\in G} \alpha(h,h^{-1}) \rho(g)X\sigma(g)^{\dag}\sigma(h^{-1})^{\dag}\\
    &= M\sigma(h)
  \end{align*}
  reveals that \(M\) is an intertwiner between the two representations.
  If they are inequivalent then, by Schur's Lemma, \(M\) must vanish, which implies for any choice of basis
  \begin{equation}
    \label{eq:74}
    \sum_{g\in G} \rho_{ij}(g)\, \overline{\sigma}_{kl}(g) = 0
  \end{equation}
  for all \(i,j,k,l\).
  If the two representations are equal then
  \begin{equation}
    \label{eq:77}
    \sum_{g\in G} \rho(g) X \rho(g)^{\dag} = \operatorname{tr}X\frac{|G|}{d_{\rho}}\mathbbm{1}_{\rho},
  \end{equation}
  which, after fixing a basis and choosing \(X = \delta_{ij}\), is what we had to show.
\end{proof}
\begin{lemma}\label{lem:g-inj-proj}
  Let \(\lambda\) be a representation of the finite group \(G\),
  whose decomposition into irreducibles contains only characters.
  In particular this means that the representation is linear.
  If this decomposition has no multiplicities greater than one, then
  \begin{equation}\label{eq:28}
    \sum_{g\in G}
    \mStk[.7cm]{{\raisebox{-0pt}{\includegraphics{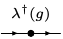}}}
                {\raisebox{-13.06595pt}{\includegraphics{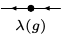}}}}
    = \sum_{g\in G}
      \raisebox{-8.58342pt}{\includegraphics{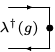}}\;
      \raisebox{-8.58342pt}{\includegraphics{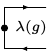}}.
  \end{equation}
\end{lemma}
\begin{proof}
  Without loss of generality we can assume that \(\lambda\) acts as a diagonal \(m\times m\) matrix,
  i.e.\ \(\lambda_{ij}(g) = \delta_{ij}\, \chi_i(g)\),
  where \(\chi_1,\ldots, \chi_m\) is a list of mutually distinct characters of \(G\).
  To show the statement we simply compare the components of the two tensors:
  \begin{align*}
    \sum_{g\in G}\,
    \mStk[.7cm]{{\scriptstyle j} {\scriptstyle l}}\,
    \mStk[.7cm]{{\raisebox{-0pt}{\includegraphics{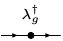}}}
                {\raisebox{-12.89926pt}{\includegraphics{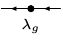}}}}\,
    \mStk[.7cm]{{\scriptstyle i} {\scriptstyle k}}
    &= \sum_{g\in G} (\lambda^{\dag}_g)_{ij}(\lambda_g)_{lk}\\[-7pt]
    &= \sum_{g\in G} \delta_{ij}\, \overline{\chi}_j(g)\, \delta_{lk} \,\chi_k(g)\\
    &= \delta_{ij}\, \delta_{jk}\, \delta_{lk}
     = \delta_{lj}\, \delta_{jk}\, \delta_{ik}\\
    &= \sum_{g\in G} \delta_{lj}\, \overline{\chi}_j(g)\, \delta_{ik} \,\chi_k(g)\\
    &= \sum_{g\in G} (\lambda^{\dag}_g)_{lj}(\lambda_g)_{ik}\\[-3pt]
    &= \sum_{g\in G}\,
    \mStk[.7cm]{{\scriptstyle j} {\scriptstyle l}}\hspace{-.1cm}
    \raisebox{-8.58342pt}{\includegraphics{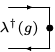}}\;
    \raisebox{-8.58342pt}{\includegraphics{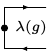}}
    \mStk[.7cm]{{\scriptstyle i} {\scriptstyle k}},
  \end{align*}
  where we used the Schur orthogonality relations, cf.\ Lemma \ref{lem:orhtogonal-characters},
  for irreducible characters twice.
  The expressions involving the Kronecker deltas are equal,
  since both of them enforce all indices to be equal.
\end{proof}
\begin{remark}
  In particular, the regular representations of abelian groups
  satisfy the assumptions of the previous lemma.
\end{remark}
\subsection{Abelian groups}
\label{sec:proj-rep-abelian-groups}
In this subsection \(G\) will always be a finite abelian group.
\begin{definition}[Pontryagin dual]
  The \emph{Pontryagin dual} \(\widehat{G}\) of a finite abelian group \(G\)
  is the group of homomorphisms from \(G\) to the unitary group \(\operatorname{U}(1)\),
  in symbols \(\widehat{G} \coloneq \operatorname{Hom}(G, \operatorname{U}(1))\).
\end{definition}
\begin{remark}
  There is always a (non-canonical) isomorphism \(G\cong \widehat{G}\)
  and for any subgroup \(H \leq G\) there is a canonical restriction map
  \(\Res ^G_H \colon \widehat{G}\rightarrow \widehat{H}\).
\end{remark}
\begin{definition}[Bicharacter]
  A \emph{bicharacter} on a group \(G\) is a map
  \(G\times G \rightarrow \operatorname{U}(1)\),
  which restricts to group homomorphisms in both arguments.
  A bicharacter \(\omega\) is called \emph{symmetric} if \(\omega(g,h) = \omega(h,g)\)
  and is called \emph{antisymmetric} if \(\omega(g,h) = \omega(h,g)^{-1}\).
\end{definition}
\begin{lemma}\label{lem:bicharacter-com}
  Let \(\omega\) be a bicharacter on \(G\) and
  construct the linear representation \(\lambda\) on \(\mathbb{C}G\) by
  \begin{equation}
    \label{eq:14}
    \lambda(g) \coloneq \sum_{h\in G} \omega(g, h) \dyad{h}{h}.
  \end{equation}
  Up to the bicharacter \(\omega\),
  the representation \(\lambda\) commutes with any projective left regular representation of \(G\):
  \begin{equation}
    \label{eq:44}
    \lambda(g) \tw{\alpha}{L}(h) = \omega(g,h) \tw{\alpha}{L}(h)\lambda(g)
  \end{equation}
\end{lemma}
\begin{proof}
  This is a simple computation after inserting Equation~\eqref{eq:15}.
\end{proof}
\begin{remark}
  It turns out that the group of antisymmetric bicharacters of G is isomorphic to 
  \(\operatorname{H}^2(G,\operatorname{U}(1))\).
\end{remark}
Given a 2-cocycle \(\alpha\in \operatorname{Z}^2(G, \operatorname{U}(1))\)
we can construct an antisymmetric bicharacter
\(c^{\alpha}\colon G\times G\rightarrow \operatorname{U}(1)\)
by defining
for all \(g,h\in G\)
\begin{equation}
  c^{\alpha}(g, h) \coloneq \frac{\alpha(g,h)}{\alpha(h,g)}.
\end{equation}
Since \(c^{\alpha}\) is clearly antisymmetric we only have to check
that it is a group homomorphism in the first argument.
So let \(g_1, g_2, h \in G\) and compute
\begin{align*}
  c^{\alpha}(g_1 g_2,h)
  &= \frac{\alpha(g_1 g_2,h)}{\alpha(h, g_1 g_2)}\\
  &= \frac{\alpha(g_1, g_2 h) \alpha(g_2, h)}{\alpha(g_1, g_2)}
     \frac{\alpha(g_1, g_2)}{\alpha(h g_1, g_2)\alpha(h, g_1)}\\
  &= \frac{\alpha(g_1 h, g_2)\alpha(g_2, h)\alpha(g_1, h)}{
     \alpha(h, g_2)\alpha(h g_1, g_2)\alpha(h, g_1)}\\
  &= \frac{\alpha(g_1, h)}{\alpha(h, g_1)}
     \frac{\alpha(g_2, h)}{\alpha(h, g_2)}\\
  &= c^{\alpha}(g_1, h) c^{\alpha}(g_2, h).
\end{align*}
\begin{remark}
  For \(\rho\) a projective representation of \(G\) with 2-cocycle
  \(\alpha\in \operatorname{Z}^2(G, \operatorname{U}(1))\)
  and for all \(g,h\in G\)
  we clearly have
  \begin{equation}
    \label{eq:7}
    \rho(h)^{-1}\rho(g)\rho(h) = \frac{\alpha(g,h)}{\alpha(h,g)}\rho(g) = c^{\alpha}(g,h) \rho(g).
  \end{equation}
\end{remark}
\begin{definition}[Slant product]\label{def:slant-product}
  Given \([\alpha]\in \operatorname{H}^2(G, \operatorname{U}(1))\), we can use the antisymmetric bicharacter
  \(c^{\alpha}\) to define the group homomorphism
  \begin{align}
    \imath\alpha \colon
    \eqmakebox[i1]{\(G\)} & \longrightarrow \eqmakebox[i2]{\(\widehat{G}\)} \\\nonumber
    \eqmakebox[i1]{\(g\)} & \longmapsto \eqmakebox[i2]{\(\imath_g\alpha\)} \coloneq c^{\alpha}(g, -).
  \end{align}
  We call \(\imath\alpha\)
  the \emph{slant product} of \([\alpha]\).
\end{definition}
\begin{remark}
  Clearly, \(\imath\alpha\) depends only on the class \([\alpha]\), not the specific representative.
  Sometimes we will write
  \(\imath_{(-)}\alpha\)
  instead of
  \(\imath\alpha\).
\end{remark}
\begin{example}\label{ex:sylvester}
  Every irreducible projective representation of \(\mathbb{Z}_n \times \mathbb{Z}_n\) 
  is linearly equivalent to a representation
  \(\sigma^{\zeta}\colon (b,u) \mapsto X^b Z^u \), where
  \begin{equation}\label{eq:25}
    X \coloneq
    \begin{pmatrix}
      0 & 0 & \cdots & 0 & 1\\
      1 & 0 & \cdots & 0 & 0\\
      0 & 1 & \cdots & 0 & 0\\
      \vdots & \vdots & \ddots & \vdots & \vdots\\
      0 & 0 & \cdots & 1 & 0
    \end{pmatrix},
  \end{equation}
  \begin{equation}\label{eq:26}
    Z \coloneq
    \begin{pmatrix}
      1 & 0 & 0& \cdots & 0\\
      0 & \zeta & 0& \cdots & 0\\
      0 & 0 & \zeta^2& \cdots & 0\\
      \vdots & \vdots & \vdots & \ddots & \vdots\\
      0 & 0 & 0& \cdots & \zeta^{n-1}
    \end{pmatrix}
  \end{equation}
  and \(\zeta\) is an \(n^{\mathrm{th}}\) root of unity.
  The matrices \(X\) and \(Z\) are known as the \emph{Sylvester shift} and \emph{clock} matrices
  and satisfy \(XZ = \zeta ZX\).
  The slant product of this representation is easily computed to be
  \begin{equation}\label{eq:37}
    (\imath_{(b_1, u_1)}\alpha^{\zeta})(b_2, u_2) = \zeta^{u_1b_2 - u_2b_1}
  \end{equation}
  for all \(b_i, u_i \in \mathbb{Z}_n\),
  where \(\alpha^{\zeta}\) denotes the 2-cocyle of~\(\sigma^{\zeta}\).
  Hence, the slant product gives a one-to-one correspondence between
  \(\operatorname{H}^2(\mathbb{Z}_n\times \mathbb{Z}_n, \operatorname{U}(1))\)
  and the set of \(n^{\mathrm{th}}\) roots of unity.
\end{example}
\begin{definition}[Degeneracy]\label{def:degeneracy}
  Mirroring the terminology of the associated bicharacter,
  \(K_{\alpha} \coloneq \operatorname{ker} \imath\alpha \leq G\)
  is called the \emph{degeneracy} of the 2-cocycle \(\alpha\) and
  if \(K_{\alpha}\) is trivial, the 2-cocycle is called \emph{non-degenerate}.
\end{definition}
\begin{remark}
  The representations of the previous example all have non-degenerate 2-cocyles.
\end{remark}
\begin{lemma}\label{lem:10}
  Let \(\rho\) be an irreducible \(\alpha\)-projective representation 
  on the vector space \(V\), then there exists a unique
  \(\lambda_{\rho} \in \widehat{K}_{\alpha}\) such that
  \(\lambda_{\rho}\cdot \operatorname{id}_V =\rho |_{K_{\alpha}} \equiv\Res ^G_{K_{\alpha}} \rho.\)
\end{lemma}
\begin{proof}
  By Schur's lemma and Equation~\eqref{eq:7} the result immediately follows.
\end{proof}
\begin{lemma}\label{lem:11}
  Let \(d\) be the rank of an irreducible \(\alpha\)-representation,
  then
  \(|G| = d^2\left| \operatorname{ker}(\imath \alpha) \right|\).
  In particular, all irreducible \(\alpha\)-representations share the same rank.
\end{lemma}
\begin{lemma}\label{lem:12}
  Two irreducible \(\alpha\)-representations \(\rho_1\) and \(\rho_2\)
  are linearly equivalent if and only if
  \(\lambda_{\rho_1} = \lambda_{\rho_2}\),
  where \(\lambda_{\rho_i} \in \widehat{K}_{\alpha}\) is defined in Lemma~\ref{lem:10}.
\end{lemma}
Lemmas~\ref{lem:11} and~\ref{lem:12} are stated and proved in e.g.\ 
Chapter 6 of the textbook \cite{berkovicCharactersFiniteGroups1998},
as Theorem 39 (b) and Lemma 44 respectively.
\begin{remark}
  The previous lemmas show that for a given 2-cocycle of a finite abelian group,
  there is a unique irreducible projective representation,
  up to projective equivalence.
  Lemma~\ref{lem:12} exhibits a bijection between irreducible linear representations of \(K_{\alpha}\)
  and irreducible \(\alpha\)-representations of \(G\), up to linear equivalence.
\end{remark}
\begin{lemma}\label{lem:13}
  Let \(r_{\alpha}\) denote
  (as in Definition~\ref{lem:4})
  the number of linearly inequivalent irreducible
  projective representations of \(G\) with 2-cocycle \(\alpha\).
  We have that
  \begin{enumerate}
  \item \(r_{\alpha} = \left| \operatorname{ker} \imath\alpha \right|\) and
  \item \(|\widehat{G}| = r_{\alpha} \left|\operatorname{ker}\Res ^G_{K_{\alpha}}\right|\).
  \end{enumerate}
\end{lemma}
\pagebreak
\begin{proof}
  ~\begin{enumerate}
  \item Clearly,
    \(\left| \operatorname{ker} \imath\alpha \right| = \operatorname{dim}_{\mathbb{C}} \operatorname{Z}(\mathbb{C}^{\alpha}G)\),
    hence, the result follows from Lemma~\ref{lem:5}.
    Alternatively, this also follows directly from Lemma~\ref{lem:12}.
  \item
    For finite abelian groups,
    linear characters on a subgroup can always be extended to the whole group.
    Hence, \(\Res ^G_{K_{\alpha}}\) is surjective
    and the result follows from Lemma~\ref{lem:12}.\qedhere
  \end{enumerate}
\end{proof}
\begin{corollary}\label{cor:can-iso}
  \(\operatorname{im}\imath\alpha=
  \operatorname{ker}\Res ^{G}_{K_{\alpha}}\).
\end{corollary}
\begin{proof}
  Clearly,
  \(\operatorname{im}\imath\alpha\)
  is a subgroup of
  \(\operatorname{ker}\Res ^{G}_{K_{\alpha}}\)
  and since
  \begin{align}\label{eq:8}
    \left| \operatorname{im} \imath\alpha \right|
    = \left| G/ \operatorname{ker}\imath\alpha \right|
    = \frac{|G|}{r_{\alpha}}
    = \left| \operatorname{ker}\Res ^{G}_{K_{\alpha}} \right|,
  \end{align}
  equality of the two groups follows.
\end{proof}
\begin{lemma}\label{lem:17}
  For any irreducible \(\alpha\)-representation \(\rho\)
  and \(\chi_1, \ldots, \chi_{r_{\alpha}}\in \widehat{G}\) such that
  \(\widehat{G}/ \operatorname{ker} \Res ^G_{K_{\alpha}}
    = \left\{ \left[\chi_1\right], \ldots, \left[\chi_{r_{\alpha}}\right] \right\}\),
  we have
  \begin{equation}\label{eq:9}
    \tw{\alpha}{L}\cong \bigoplus_{i=1}^{r_{\alpha}} \chi_i \mathbbm{1}_{d_{\rho}} \otimes \rho.
  \end{equation}
\end{lemma}
\begin{proof}
  Lemmas~\ref{lem:12} and~\ref{lem:13}
  show that \(\chi_1\rho, \ldots, \chi_{r_{\alpha}}\rho\)
  is a complete list of
  pairwise linearly inequivalent,
  irreducible \(\alpha\)-projective representations of \(G\) and
  Corollary~\ref{cor:character-multiplicity} shows that
  each irreducible representation appears in the regular representation,
  with multiplicity equal to its degree.
\end{proof}
\begin{remark}
  In the case where \(\alpha\) is the trivial 2-cocycle,
  the previous formula reduces to the standard decomposition into characters.
\end{remark}
\end{multicols}
\section{Basis transformations}
\label{sec:basis-transf}
\begin{align*}
  \mathbb{A}^i
  &= \sum_{g\in G} \dyad{g}{g} \otimes U(g)^i_j A^j \\
  &= \sum_{g\in G} \dyad{g}{g} \otimes V(g)^{-1} A^i V(g)\\
  &= \left(\sum_{g\in G} \dyad{g}{g} \otimes V(g)^{-1}\right)
     \left(\mathbbm{1}_n \otimes A^i\right)
     \left(\sum_{h\in G} \dyad{h}{h} \otimes V(h)\right)\\
  &= S^{-1}\left(\mathbbm{1}_n \otimes A^i\right) S\\
  \\
  S \twb{\tau}{\mathbb{B}}^g S^{-1}
  &= S \left( \tw{\tau}{L}(g) \otimes \mathbbm{1}_{d_{\text{v}}} \right) S^{-1}\\
  &= \sum_{h_1, h_2\in G} \dyad{h_1}{h_1}\tw{\tau}{L}(g)\dyad{h_2}{h_2} \otimes V(h_1) V(h_2)^{-1}\\
  &= \sum_{h_1, h_2\in G} \tau(g,h_2)\ket{h_1}\mkern-6mu\braket{h_1|gh_2}\mkern-6mu\bra{h_2}\otimes V(h_1) V(h_2)^{-1}\\
  &= \sum_{h\in G} \tau(g,h)\dyad{gh}{h} \otimes V(gh) V(h)^{-1}\\
  &= \sum_{h\in G} \frac{\tau(g, h)}{\alpha(g, h)}\dyad{gh}{h}\otimes V(g)\\
  &= \tw{\tau\overline{\alpha}}{L}(g)\otimes V(g)\\
  \\
  S\left(\tw{\overline{\tau}}{R}(g)\otimes V(g)\right)S^{-1}
  &= \sum_{h_1, h_2\in G}
    \dyad{h_1}{h_1}\tw{\overline{\tau}}{R}(g)\dyad{h_2}{h_2} \otimes V(h_1)V(g) V(h_2)^{-1}\\
  &= \sum_{h_1, h_2\in G} \overline{\tau}(h_2g^{-1}, g)
    \ket{h_1}\mkern-6mu\braket{h_1|h_2g^{-1}}\mkern-6mu\bra{h_2}\otimes V(h_1)V(g) V(h_2)^{-1}\\
  &= \sum_{h\in G} \overline{\tau}(hg^{-1}, g) \dyad{h g^{-1}}{h} \otimes V(h g^{-1})V(g) V(h)^{-1}\\
  &= \sum_{h\in G} \frac{\alpha(hg^{-1}, g)}{\tau(hg^{-1}, g)} \dyad{h g^{-1}}{h} \otimes \mathbbm{1}\\
  &= \sum_{h\in G} \frac{\alpha(h, g)}{\tau(h, g)} \dyad{h}{hg} \otimes \mathbbm{1}\\
  &= \tw{\alpha\overline{\tau}}{R}(g) \otimes \mathbbm{1}
\end{align*}
\begin{multicols}{2}
\section{MPS Ansatz}
\label{sec:ansatz}
\begin{proof}[Proof of Lemma~\ref{lem:main2}]
  Under the assumptions that the symmetry group \(G\) factors into broken and unbroken parts
  and that \(\ket{\Psi(A)}\) does not exhibit long range entanglement not protected by \(G\),
  the results in~\cite{schuchClassifyingQuantumPhases2011} readily yield
  the existence of a canoncial form
  \begin{equation}
    \label{eq:100}
    A^i = \bigoplus_{k=1}^{n_{\mathrm{b}}} A^i_k,
  \end{equation}
  where \(n_{\mathrm{b}} \coloneq |G_{\mathrm{b}}|\)
  and the representation of \(G\) on the virtual Hilbert space is given by
  \begin{equation}
    \label{eq:97}
    (\mathcal{V}, V) = (\mathbb{C}G_{\mathrm{b}}\otimes \mathcal{W}, L\otimes W).
  \end{equation}
  Here, the regular representation \(L\) permutes between the blocks,
  while \(W\) acts within the blocks:
  \begin{equation}\label{eq:105}
    \raisebox{-13.94629pt}{\includegraphics{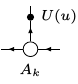}} = 
    \raisebox{-13.94629pt}{\includegraphics{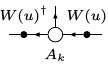}} 
  \end{equation}
  for all \(u\in G_{\mathrm{u}}\) and \(k=1, \dots, n_{\mathrm{b}}\).
  Thus, we have shown~(i).
  What remains is to show the \(G_{\mathrm{b}}\)-injectivity of \(A\).
  For simplicity, we assume that the component tensors are properly normalized,
  which implies, cf.~\cite{ciracMatrixProductDensity2017}, that each transfer operator
  \begin{equation}
    \label{eq:106}
    \mathbb{T}_k\coloneq\raisebox{-25.32742pt}{\includegraphics{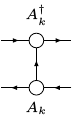}},
  \end{equation}
  \(k = 1,\ldots, n_{\mathrm{b}}\), has a unique eigenvector
  \(v_k \in \mathcal{W}^{*}\otimes \mathcal{W}\) with eigenvalue 1,
  \begin{equation}
    \label{eq:102}
    \raisebox{-25.32742pt}{\includegraphics{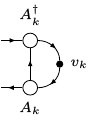}} = 
    \raisebox{-9.08113pt}{\includegraphics{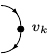}},
  \end{equation}
  and that all other eigenvalues of \(\mathbb{T}_k\) are less than 1.
  Let \(\epsilon_k < 1\) denote the the second largest eigenvalue of \(\mathbb{T}_k\)
  and let \(\epsilon\) denote the maximum of \(\left\{ \varepsilon_k \right\}\).
  Since the transfer operator \(\mathbb{T}_k\) is \(W^{\dag} \otimes W\) symmetric,
  in the sense of
  \begin{equation}
    \label{eq:101}
    \raisebox{-24.44708pt}{\includegraphics{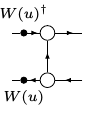}} =
    \raisebox{-24.44708pt}{\includegraphics{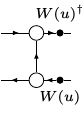}}
  \end{equation}
  for all \(u \in G_{\mathrm{u}}\),
  its unique eigenvector \(v_k\) of maximal eigenvalue also has to be symmetric,
  which implies
  \begin{equation}
    \label{eq:103}
    \sqrt{v_k}\,W(u) = W(u) \sqrt{v_k}
  \end{equation}
  for all \(u\in G_{\mathrm{u}}\) and \(k = 1,\ldots, n_{\mathrm{b}}\).
  This means that we can find a basis transformation of \(\mathcal{W}\)
  such that the mixed transfer operators satisfy
  \begin{equation}
    \label{eq:104}
    \left( \raisebox{-25.32742pt}{\includegraphics{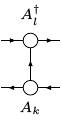}} \right)^n
    = \frac{\delta_{kl}}{d_{\mathcal{W}}}\;
    \raisebox{-9.08113pt}{\includegraphics{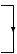}}\; \raisebox{-9.08113pt}{\includegraphics{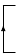}}
    + \mathcal{O}(\epsilon^n)
  \end{equation}
  for all \(k,l \in 1, \ldots, n_{\mathrm{b}}\) and \(n\in \mathbb{N}\),
  without modifying the form of the virtual representation in Equation~\eqref{eq:105}.

  Because all tensors \(A_1, \dots, A_{n_{\mathrm{b}}}\)
  are normal and linearly independent (since they generate pairwise orthogonal states)
  the matrices commuting with all \(A^i\) are precisely those of the form \(D\otimes \mathbbm{1}\),
  where \(D\) is a diagonal matrix.
  In particular, the representation \(\lambda \otimes \mathbbm{1}\) will commute with \(A\), where
  \begin{equation}\label{eq:47}
    \lambda(b) \coloneq \sum_{c\in G_{\mathrm{b}}} \omega(b,c) \dyad{c}{c}
  \end{equation}
  for a fixed non-degenerate bicharacter \(\omega\) of \(G_{\mathrm{b}}\).
  Notice that for an appropriately ordered list \(\chi_1,\ldots, \chi_{n_{\mathrm{b}}}\)
  of all irreducible characters of \(G_{\mathrm{b}}\),
  \(\lambda\) is of the form
  \begin{equation}
    \label{eq:48}
    \lambda(b) = 
    \begin{pmatrix}
      \chi_1(b)&&\\
      &\ddots&\\
      &&\chi_{n_{\mathrm{b}}}(b)
    \end{pmatrix},
  \end{equation}
  which implies that the matrices \(\left\{ \lambda(b) \mid b\in G_{\mathrm{b}} \right\}\)
  span the space of diagonal matrices of size \(n_{\mathrm{b}}\).
  Defining the projectors
  \begin{equation}
    P_{\chi} \coloneq \frac{1}{n_{\mathrm{b}}}
                \sum_{b\in G_{\mathrm{b}}}
                \chi(b)\lambda(b),
  \end{equation}
  where \(\chi\) is an irreducible character of \(G_{\mathrm{b}}\),
  allows us to write \(A\) as
  \begin{equation}
    \raisebox{-15.03612pt}{\includegraphics{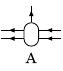}} =
    \sum_{k = 1}^{n_{\mathrm{b}}}
    \Stk{\raisebox{-13.94629pt}{\includegraphics{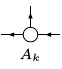}}}{\raisebox{-13.59927pt}{\includegraphics{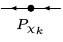}}},
  \end{equation}
  where we swapped top and bottom legs to make the following diagrams easier to draw.
  Using the orthogonality relations for irreducible characters we compute
  \begin{equation}
    \sum_{\chi\in \widehat{G_{\mathrm{b}}}}
    \Stk{\raisebox{-0pt}{\includegraphics{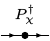}}}
        {\raisebox{-12.8215pt}{\includegraphics{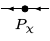}}}
        = \frac{1}{n_{\mathrm{b}}^2}
        \!\!\sum_{\substack{\chi\in \widehat{G_{\mathrm{b}}}\\ a,b\in G_{\mathrm{b}}}}\!\!
    \Stk{\raisebox{-0pt}{\includegraphics{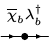}}}
        {\raisebox{-12.28815pt}{\includegraphics{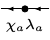}}}
     = \frac{1}{n_{\mathrm{b}}}\sum_{b\in G_{\mathrm{b}}} 
    \Stk{\raisebox{-0pt}{\includegraphics{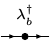}}}
        {\raisebox{-11.98816pt}{\includegraphics{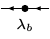}}},
  \end{equation}
  yielding for \(n\in \mathbb{N}\)
  \begin{align*}
    \left(\!\raisebox{-26.37843pt}{\includegraphics{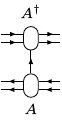}}\!\right)^n
    &= \eqmakebox[lem-2-prf]{\(\displaystyle\frac{1}{d_{\mathcal{W}}}\sum_{k = 1}^{n_{\mathrm{b}}}\)}
    \mStk[.47cm]{{\raisebox{-0pt}{\includegraphics{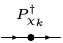}}}
                 {\raisebox{-7.08945pt}{\includegraphics{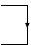}}\hspace{.4mm}
                  \raisebox{-7.08945pt}{\includegraphics{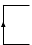}}}
                 {\raisebox{-13.59927pt}{\includegraphics{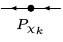}}}}
      + \mathcal{O}(\epsilon^n) \\
    &\equiv \eqmakebox[lem-2-prf]{\(\displaystyle\frac{1}{d_{\mathcal{W}}}\sum_{\chi\in \widehat{G_{\mathrm{b}}}}\)}
    \mStk[.47cm]{{\raisebox{-0pt}{\includegraphics{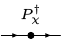}}}
                {\raisebox{-7.08945pt}{\includegraphics{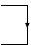}}\hspace{.4mm}
                 \raisebox{-7.08945pt}{\includegraphics{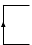}}}
                {\raisebox{-12.8215pt}{\includegraphics{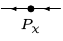}}}}
      + \mathcal{O}(\epsilon^n) \\
    &= \eqmakebox[lem-2-prf]{\(\displaystyle\frac{1}{n_{\mathrm{b}}d_{\mathcal{W}}}\sum_{b \in G_{\mathrm{b}}}\)\,}
    \mStk[.47cm]{{\raisebox{-0pt}{\includegraphics{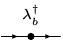}}}
                 {\raisebox{-7.08945pt}{\includegraphics{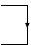}}\hspace{.4mm}
                  \raisebox{-7.08945pt}{\includegraphics{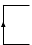}}}
                 {\raisebox{-11.98816pt}{\includegraphics{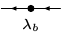}}}}
      + \mathcal{O}(\epsilon^n),
  \end{align*}
  which is, after using Lemma~\ref{lem:g-inj-proj}, precisely what we had to show.
\end{proof}
\end{multicols}
\end{appendices}

\section*{References}
\begin{multicols}{2}
\printbibliography[heading=none]
\end{multicols}
\end{document}